\documentclass[twoside]{article}

\usepackage[accepted]{aistats2022}
\usepackage[round]{natbib}

\usepackage{algorithm}
\usepackage{algpseudocode}

\usepackage[utf8]{inputenc} 
\usepackage[T1]{fontenc}    
\usepackage{hyperref}       

\hypersetup{
  colorlinks   = true, 
  urlcolor     = blue, 
  linkcolor    = blue, 
  citecolor   = blue 
}

\usepackage{url}            
\usepackage{booktabs}       
\usepackage{amsfonts}     
\usepackage{amsthm}     
\usepackage{nicefrac}       
\usepackage{microtype}      
\usepackage{xcolor}

\usepackage{graphicx} 
\usepackage{subfigure} 

\usepackage{amssymb}
\usepackage{amsmath}
\usepackage{amsfonts}
\usepackage{makecell}

\DeclareMathOperator{\argmin}{\arg\min}

\newtheorem{theorem}{Theorem}
\newtheorem{lemma}{Lemma}

\begin{document}

%
\runningtitle{Learning Personalized Item-to-Item Recommendation Metric}

%
\runningauthor{Hoang, Deoras, Zhao, Li, Karypis}

\twocolumn[

\aistatstitle{Learning Personalized Item-to-Item Recommendation\\ Metric via Implicit Feedback}

\aistatsauthor{\quad\quad\quad\quad Trong Nghia Hoang, Anoop Deoras \And \quad\quad\quad\quad\quad Tong Zhao, Jin Li \And \quad\quad\quad\quad George Karypis}

\aistatsaddress{\quad\quad\quad\quad AWS AI, Amazon \And \quad\quad\quad\quad\quad Personalization, Amazon \And \quad\quad\quad\quad\quad AWS AI, Amazon} ]

\begin{abstract}
\vspace{-3mm}
This paper studies the item-to-item recommendation problem in recommender systems from a new perspective of metric learning via implicit feedback. We develop and investigate a personalizable deep metric model that captures both the internal contents of items and how they were interacted with by users. There are two key challenges in learning such model. First, there is no explicit similarity annotation, which deviates from the assumption of most metric learning methods. Second, these approaches ignore the fact that items are often represented by multiple sources of meta data and different users use different combinations of these sources to form their own notion of similarity.

\noindent To address these challenges, we develop a new metric representation embedded as kernel parameters of a probabilistic model. This helps express the correlation between items that a user has interacted with, which can be used to predict user interaction with new items. Our approach hinges on the intuition that similar items induce similar interactions from the same user, thus fitting a metric-parameterized model to predict an implicit feedback signal could indirectly guide it towards finding the most suitable metric for each user. To this end, we also analyze how and when the proposed method is effective from a theoretical lens. Its empirical effectiveness is also demonstrated on several real-world datasets.\vspace{-3mm}
\end{abstract}

\section{Introduction}
\label{sec:intro}\vspace{-2mm}
Item recommendation is one of the fundamental tasks in a recommender system which is applicable to many scenarios such as \emph{you may also like} on e-commerce platforms (e.g., Amazon, Alibaba) or \emph{because you watched} on content streaming services (e.g., Netflix). These include two specific use cases of (1) user-centric and (2) item-centric recommendations. In user-centric recommendation, the focus is on recommending items that fit best with a profile of a target user. In item-centric recommendation, the focus is instead on recommending items that are similar to a target item.

\noindent To date, most recommendation methods have focused on the user-centric context of recommending relevant items to a target user. However, such user-centric methods are often not suitable for item-centric use cases since (as mentioned above) their focus is to find items that fit best with a user's profile but might not necessarily be similar to the target item that the user is currently interested in. To the best of our knowledge, there are very few works devised to tackle item-item recommendation directly. Most notable works among those are \emph{sparse linear method} (SLIM) \citep{Karypis2011}, which was adapted from user-centric collaborative filtering (CF) methods \citep{he2017neural,Salakhutdinov2007,Karypis2011,Salakhutdinov2008,Yu2009}, and \emph{semi-parametric embedding} (SPE) \citep{Li2019}, which combines elements of both CF- and content-based methods \citep{Lops2011,Roy2000,Billsus2007}. But, SLIM \citep{Karypis2011} does not make use of meta information while \emph{semi-parametric embedding} (SPE) \citep{Li2019} ignores different similarity notions that we mentioned above. In both cases, there is a need to develop a personalizable item-to-item distance metric that not only capture the similarities between items across different sources of meta data but also \emph{how} these are perceived by different users.

\noindent This leads us to the problem of metric learning \citep{Weinshall2003,Davis2005,Kulis2012,Lebanon2003,Xing2003} which aims to learn a distance measure on the feature space of items that can capture well the semantic similarities of items on their original input space \citep{Chopra2005,Goldberger2005,Zhou2017,Yang2018} (see Section~\ref{sec:related}). However, most (deep) metric learning methods were developed outside the context of a practical recommendation system \citep{guo2017deepfm,he2017neural,koren2009matrix,Karypis2011,sedhain2015autorec} where class labels or even co-view signals of commercial items are not fine-grained enough to determine whether two items are similar or not. For example, two movies might belong to the same genre but they are not considered similar due to other traits such as cast, producer and plot. On the other hand, co-viewed movies might be driven by an exploration behavior rather than by their intrinsic similarities \citep{Mary18}. This invalidates vanilla application of existing supervised metric learning algorithms. 

Furthermore, users tend to have different preferences in forming their own notion of similarity from different meta information channel of commercial items. For example, some users would consider movies from the same genres or thematic content to be similar, while others might prefer movies from specific cast. This invalidates direct uses of unsupervised methods seeking to preserve a single geometry of item-to-item metric since there can be as many as the no. of different users.

\noindent {\bf Motivation.} This motivates us to develop a multi-channel metric representation that can be learned via implicit feedback on how good it is towards a downstream prediction task. To achieve this, we first propose and investigate an ensemble construction of multiple Siamese Twin segments \citep{Chopra2005} such that each segment comprises two identical towers that capture the embedded similar or dissimilar traits between two input items for each source of meta data describing a certain aspect of their internal contents. As there is no direct feedback to train this metric ensemble, we use a surrogate prediction task to provide a self-supervised feedback for training, which can also be used to personalize the metric for a user.

\noindent {\bf Key Idea $\&$ Contributions.} This is achieved by embedding the ensemble representation as part of the kernel parameters expressing the correlation between items within a prediction model, which is fitted to predict how a user would interact with an item based on its correlation with previous items that the user has interacted with. Our approach hinges on the intuition that similar items would induce similar interactions from the same user, thus learning an interaction (e.g., rating) prediction model based on the metric representation can implicitly guide it towards capturing the right metric for each user. In particular, we contribute:

\noindent {\bf 1.} An adapted Gaussian process (GP) \citep{Rasmussen06} regression model whose kernel function is parameterized by an ensemble of Siamese Twin segments (Section~\ref{sec:method-ssl}). The GP model is fitted to predict the average user rating of an unseen item given items with observed ratings. As the correlation is expressed in terms of metric representation, a well-fitted GP would be encouraged to find a well-behaved metric that correctly preserves the averaged similarity geometry of items.

\noindent {\bf 2.} A personalization scheme that warps the averaged similarity geometry of items into a personalized one that better fits each specific user (Section~\ref{sec:method-personalize}) via optimizing the combining coefficient of the ensemble. This makes sense since the content extracted from different meta-data channels is user-agnostic, leaving only the combining parameters user-dependent.

\noindent {\bf 3.} A theoretical analysis (Section~\ref{sec:theory}) that analyzes the effectiveness of the proposed self-supervised metric learning algorithm in terms of the statistical relevance between the surrogate prediction task (e.g., rating prediction) and the true (unknown) item-to-item metric. Our results (Theorems~\ref{theo:1} and~\ref{theo:2}) show that under reasonable assumptions, the learned metric is close to the true metric with high probability if there is a sufficient amount of observations from the surrogate function.

\noindent {\bf 4.} An empirical evaluation of the proposed method on an experimentation benchmark comprising the several public datasets such as the MovieLens \citep{harper2015movielens} and Yelp Review datasets (Section~\ref{sec:exp}).

\section{Related Work}
\label{sec:related}

\subsection{Metric Learning and Siamese Network}
\label{sec:siamese}
One prominent line of research in metric learning focuses on supervised methods \citep{Weinshall2003,Domeniconi2002,Hastie1996,Tsang2003,Lebanon2003,Xing2003,Kwok2005,Yeung2003} which assume there exist training examples of similar and dissimilar items are available. The metric learning task is thus reduced to learning a scoring function that pushes down on similar pairs while pushing up on dissimilar pairs. One notable example of such metric learning method is the Siamese network which can be learned via optimizing a contrastive loss.

{\bf Siamese Network.} As developed in \citep{Chopra2005}, Siamese network has a two-tower architecture that was specifically devised for contrastive learning. In a nutshell, a Siamese net is expected to accept a pair of input items $(\mathbf{x}_a,\mathbf{x}_b)$ and output a numeric distance between them or a probability that they are dissimilar. For a pair of similar items, we expect this distance or probability to be small and conversely, for dissimilar items, we expect it to be above a certain margin.

To achieve this, the Siamese net has two identical network segments, $\mathbf{F}_a(\mathbf{x}; \boldsymbol{\gamma}) \equiv \mathbf{F}(\mathbf{x}; \boldsymbol{\gamma})$ and $\mathbf{F}_b(\mathbf{x}; \boldsymbol{\gamma}) \equiv \mathbf{F}(\mathbf{x}; \boldsymbol{\gamma})$, whose outputs, $\mathbf{z}_a = \mathbf{F}(\mathbf{x}_a; \boldsymbol{\gamma})$ and $\mathbf{z}_b = \mathbf{F}(\mathbf{x}_b; \boldsymbol{\gamma})$ reside in a metric space $\mathbb{R}^p$ equipped with a parameterized distance $\mathbf{D}(\mathbf{z}_a,\mathbf{z}_b) = (\mathbf{z}_a - \mathbf{z}_b)^\top\boldsymbol{\Lambda}(\mathbf{z}_a - \mathbf{z}_b)$ where $\boldsymbol{\Lambda} = \mathrm{diag}[\lambda_1, \lambda_2, \ldots, \lambda_p]$. Thus, given a pair $(\mathbf{x}_a,\mathbf{x}_b)$, the output of the Siamese net is $\mathbf{D}(\mathbf{x}_a,\mathbf{x}_b) = $
\begin{eqnarray}
\hspace{-0.5mm}\left(\mathbf{F}\Big(\mathbf{x}_a; \boldsymbol{\gamma}\Big) \hspace{-0.5mm}-\hspace{-0.5mm} \mathbf{F}\Big(\mathbf{x}_b; \boldsymbol{\gamma}\Big)\right)^{\hspace{-1mm}\top} \hspace{-1mm}\boldsymbol{\Lambda} \left(\mathbf{F}\Big(\mathbf{x}_a; \boldsymbol{\gamma}\Big) \hspace{-0.5mm}-\hspace{-0.5mm} \mathbf{F}\Big(\mathbf{x}_b; \boldsymbol{\gamma}\Big)\right)\label{eq:4}
\end{eqnarray}
Then, suppose training examples $((\mathbf{x}^{i}_a,\mathbf{x}^{i}_b), y^i_{ab})_{i=1}^n$ are available where $y^i_{ab} = 0$ indicate $(\mathbf{x}_a^{i}, \mathbf{x}_b^{i})$ is a pair of similar items and otherwise for $y^i_{ab} = 1$. The parameterization of the metric net, $\boldsymbol{\gamma}$ and $\boldsymbol{\Lambda}$, can be learned via optimizing the following contrastive loss,
\begin{eqnarray}
\hspace{-6.5mm}\mathbf{L}\Big(\boldsymbol{\gamma}, \boldsymbol{\Lambda}\Big)&=&\sum_{i=1}^n \Bigg[\Big(1 - y^i_{ab}\Big) \mathbf{D}\left(\mathbf{x}^{i}_a, \mathbf{x}^{i}_b\right)\Bigg]\nonumber \\
&+& \sum_{i=1}^n \Bigg[y^i_{ab} \max\Big(0, \tau - \mathbf{D}\left(\mathbf{x}^{i}_a, \mathbf{x}^{i}_b\right)\Big)\Bigg]\label{eq:5}
\end{eqnarray}
where $\tau$ is a contrastive margin such that the distance for dissimilar pairs are encouraged to be increased up to $\tau$ but not more than that. This implies forcing the distance between a dissimilar pair to be more than $\tau$ only yields diminishing gain in improving the discriminative capacity of the model. Thus, the loss is zeroed out in such cases to focus on minimizing the gap between the other similar pairs.

\subsection{Gaussian Processes}
\label{sec:gp}
A Gaussian process \citep{Mackay1998IntroductionTG} defines a probabilistic prior over a random function $g(\mathbf{x})$. This prior is in turn defined by a mean function $m(\mathbf{x}) = 0$\footnote{For simplicity, we assume a zero mean function since we can always re-center the training outputs around $0$.} and a kernel function $k(\mathbf{x}, \mathbf{x}')$. Such prior stipulates that for an arbitrary finite subset of inputs  $\{\mathbf{x}_1, \mathbf{x}_2, \ldots, \mathbf{x}_n\}$, the corresponding output vector $\mathbf{g} = [g(\mathbf{x}_1)\ g(\mathbf{x}_2) \ldots\  g(\mathbf{x}_n)]^\top$ is distributed by a multivariate Gaussian, $\mathbf{g} \sim \mathbb{N}(\mathbf{0}, \mathbf{K})$. 

Here, the entries of the covariance matrix $\mathbf{K}$ are computed using the aforementioned kernel function. That is, $\mathbf{K}_{ij} = k(\mathbf{x}_i, \mathbf{x}_j)$ where common examples of $k(\mathbf{x}_i,\mathbf{x}_j)$ are detailed in \citep{Rasmussen06} but in practice, the exact choice of the kernel function usually depends on the application. To predict with GP, let $\mathbf{x}_\ast$ be an unseen input whose corresponding output $g_\ast = g(\mathbf{x}_\ast)$ we wish to predict. Then, assuming a noisy setting where we only observe a noisy observation $r(\mathbf{x}) \sim \mathbf{N}(g(\mathbf{x}), \sigma^2)$ instead of $g(\mathbf{x})$ directly, the predictive distribution of $g_\ast$ is:
\begin{eqnarray}
\hspace{-2mm}g_\ast \ \triangleq\ g(\mathbf{x}_\ast) \mid \mathbf{r} &\sim& \mathbb{N}\Big(\mathbf{k}_\ast^\top(\mathbf{K} + \sigma^2\mathbf{I})^{-1}\mathbf{r},\  k(\mathbf{x}_\ast,\mathbf{x}_\ast) \nonumber\\
&-& \mathbf{k}_\ast^\top(\mathbf{K} + \sigma^2\mathbf{I})^{-1}\mathbf{k}_\ast\Big) \ ,\label{eq:2}
\end{eqnarray}
where $\mathbf{r} = [r(\mathbf{x}_1) \ldots r(\mathbf{x}_n)]^\top$. The defining parameter $\boldsymbol{\psi}$ of $k(\mathbf{x},\mathbf{x}')$ is crucial to the predictive performance and needs to be optimized via minimizing the negative log likelihood (NLL) of $\mathbf{r}$,
\begin{eqnarray}
\hspace{-8mm}\ell(\boldsymbol{\psi}) \hspace{-2mm}&\propto&\hspace{-2mm} \frac{1}{2}\log \Big|\mathbf{K}_{\boldsymbol{\psi}} + \sigma^2\mathbf{I}\Big| + \frac{1}{2}\mathbf{r}^\top\Big(\mathbf{K}_{\boldsymbol{\psi}} + \sigma^2\mathbf{I}\Big)^{-1}\mathbf{r} \label{eq:3}
\end{eqnarray}
In the above, we use the subscript $\boldsymbol{\psi}$ to indicate that $\mathbf{K}$ is parameterized by $\boldsymbol{\psi}$ which, in our case, 
is the parameterization of a Siamese network (Section~\ref{sec:siamese}). In practice, both training $\boldsymbol{\Theta}$ and prediction incur $\mathbf{O}(n^3)$ cost. For better scalability, there have been numerous developments on sparse GPs \citep{Titsias09,Miguel10,Hensman13,NghiaICML15,NghiaICML16,Hoang2020} whose computation are only linear in $n$ (see Appendix~\ref{app:e}).

\section{Metric Learning via Gaussian Process with Siamese Kernel}
\label{sec:method}
We will formalize our intuition (Section~\ref{sec:intro}) of self-supervised learning (SSL) of item-to-item metric, namely the \emph{SSL metric}, in Section~\ref{sec:method-ssl}. We will then show how such \emph{SSL metric} can also be personalized for each user via minimizing a new loss function as proposed in Section~\ref{sec:method-personalize}. 

\subsection{Self-Supervised Metric Learning with Gaussian Processes}
\label{sec:method-ssl}
Let $r(\mathbf{x})$ denote a related prediction target of item $\mathbf{x}$ whose training examples $\{(\mathbf{x}_i, r(\mathbf{x}_i))\}_{i=1}^n$ are readily available from our data. For example, $r(\mathbf{x})$ can be an averaged review score for $\mathbf{x}$, which aggregates the ratings given to $\mathbf{x}$ by the users who interacted with it. Our goal is to build a prediction model such that for an unseen item $\mathbf{x}_\ast$, its prediction $\widehat{r}(\mathbf{x}_\ast)$ is largely based on its metric-based correlation with the training items $\mathbf{x}_1, \mathbf{x}_2, \ldots, \mathbf{x}_n$. This is the standard prediction pattern of kernel-based methods such as Gaussian process (GP). To substantiate this, we parameterize the kernel function of the GP prior using the aforementioned Siamese network (Section~\ref{sec:siamese}) as detailed below,
\begin{eqnarray}
k\Big(\mathbf{x}_a, \mathbf{x}_b; \boldsymbol{\psi}\Big) &=& \mathrm{exp}\left(-\frac{1}{2}\mathbf{D}\left(\mathbf{x}_a, \mathbf{x}_b\right)\right)\label{eq:6}
\end{eqnarray}
where $\mathbf{D}(\mathbf{x}_a,\mathbf{x}_b)$ is defined in Eq.~\eqref{eq:4}. Here, the kernel parameterization $\boldsymbol{\psi} = \{\boldsymbol{\Lambda},\boldsymbol{\gamma}\}$ consists of two parts. First, $\boldsymbol{\gamma}$ denotes the defining parameters of the network segment that maps $\mathbf{x}$ to a vector in a metric space. Second, $\boldsymbol{\Lambda}$ specifies the correlation and unit scales across different dimensions of the metric space. For example, if we choose $\boldsymbol{\Lambda}$ to be the diagonal matrix, then the dimensions of the metric space are uncorrelated and their unit scales are the elements on the diagonal of $\boldsymbol{\Lambda}$. Then, given training examples $\{(\mathbf{x}_i, r_i)\}_{i=1}^n$ where $r_i \sim \mathbb{N}(r(\mathbf{x}_i), \sigma^2)$ are the noisy observations of $r(\mathbf{x}_i)$ perturbed with Gaussian noises, the metric parameters can be optimized via minimizing
\begin{eqnarray}
\hspace{-6.5mm}\ell(\boldsymbol{\psi}) \hspace{-2mm}&=&\hspace{-2mm} \frac{1}{2}\log \Big|\mathbf{K}_{\boldsymbol{\psi}} + \sigma^2\mathbf{I}\Big| + \frac{1}{2}\mathbf{r}^\top\Big(\mathbf{K}_{\boldsymbol{\psi}} + \sigma^2\mathbf{I}\Big)^{-1}\mathbf{r} \label{eq:7}
\end{eqnarray}
with respect to $\boldsymbol{\psi}$ and $\sigma$ where $\mathbf{r} = [r_1\ r_2 \ldots r_n]^\top$ and $\boldsymbol{\psi} \triangleq \{\boldsymbol{\Lambda}, \boldsymbol{\gamma}\}$. Here, Eq.~\eqref{eq:7} is the same as Eq.~\eqref{eq:3} except for that entries of $\mathbf{K}_{\boldsymbol{\psi}}$ were computed by Eq.~\eqref{eq:6} above. We will demonstrate later in Section~\ref{sec:exp} that training $\{\boldsymbol{\Lambda},\boldsymbol{\gamma}\}$ using Eq.~\eqref{eq:7} is more effective than fitting them using Eq.~\eqref{eq:5} which requires direct feedback that cannot be acquired without incurring considerable label noise. 

{\bf Key Result:} We will also show in Section~\ref{sec:theory} below (see Theorem~\ref{theo:1} and Theorem~\ref{theo:2}) that assuming a certain statistical relationship (see {\bf A1-A3}) between $\mathbf{r}$ and the true metric $\mathbf{D}_\ast(\mathbf{x},\mathbf{x}')$, the learned \emph{SSL metric} $\mathbf{D}(\mathbf{x}, \mathbf{x}')$ via fitting Eq.~\eqref{eq:7} is arbitrarily close to $\mathbf{D}_\ast(\mathbf{x}, \mathbf{x}')$ if the number of observations $n$ is sufficiently large.

\subsection{Learning Personalizable Metric with Gaussian Processes}
\label{sec:method-personalize}
Furthermore, to account for multiple different sources of meta data describing an item that lead to different user preferences over their uses in combination, we first extend the above Siamese network architecture into an ensemble construction of multiple Siamese segments. 

{\bf Siamese Ensemble.} Let $\mathbf{x} = (\mathbf{x}^{(1)},\mathbf{x}^{(2)}, \ldots, \mathbf{x}^{(p)})$ denote its multi-view representation across $p$ different meta channels where $\mathbf{x}^{(i)}$ denote $\mathbf{x}$'s description in channel $i$. The Siamese ensemble is $\mathbf{D}\left(\mathbf{x}_a, \mathbf{x}_b\right) =$
\begin{eqnarray}
\sigma\Bigg(\left[\mathbf{D}_1\left(\mathbf{x}^{(1)}_a, \mathbf{x}^{(1)}_b\right), \ldots, \mathbf{D}_p\left(\mathbf{x}_a^{(p)}, \mathbf{x}^{(p)}_b\right)\right]^\top \mathbf{h} + b\Bigg) \label{eq:8}
\end{eqnarray}
where $\sigma(z) = 1 / (1 + \mathrm{exp}(-z))$ denotes the sigmoid function while $\mathbf{w} = (\mathbf{h}, b)$ where $\mathbf{h} \in \mathbb{R}^p$ and $b \in \mathbb{R}$ are learnable weights that aggregate the individual Siamese distances into a single distance metric. In our personalized context, the metric computation for each user $u$ shares the same set of Siamese individual distance functions (and their parameterizations) but differs in how these individual Siamese distances were combined via different choices of $\mathbf{w} \leftarrow \mathbf{w}_u = (\mathbf{h}_u, b_u)$.

{\bf Learning Personalizable Metric.} First, we observe that the parameterization $\boldsymbol{\psi}_i = (\boldsymbol{\gamma}_i,\boldsymbol{\Lambda}_i)$ of each Siamese segment $\mathbf{D}_i$ is user-agnostic since the intrinsic similarities between items across single channels are not user-dependent. The personalization must therefore concern only ensemble parameterization $\mathbf{w} \leftarrow \mathbf{w}_u$. This raises the question of how can we build a personalizable parameterization which can be fast adapted to an arbitrary user with limited personal data?

To address this question, one ad-hoc choice is to reuse the self-supervised learning recipe in Section~\ref{sec:method-ssl} to optimize for $\mathbf{w} = (\mathbf{h}, b)$, which can be achieved by re-configuring the kernel function $k(\mathbf{x}, \mathbf{x}')$ in Eq.~\eqref{eq:6} with the ensemble distance $\mathbf{D}$ in Eq.~\eqref{eq:8} and minimizing the following NLL loss $\ell(\boldsymbol{\psi},\mathbf{w}) =$
\begin{eqnarray}
\hspace{-2mm}\frac{1}{2}\log \Big|\mathbf{K}_{(\boldsymbol{\psi}, \mathbf{w})} + \sigma^2\mathbf{I}\Big| +  \frac{1}{2}\mathbf{r}^\top\Big(\mathbf{K}_{(\boldsymbol{\psi},\mathbf{w})} + \sigma^2\mathbf{I}\Big)^{-1}\mathbf{r} \ .\label{eq:m1}
\end{eqnarray}
with respect to $\mathbf{w}$, while fixing $\boldsymbol{\psi} = (\boldsymbol{\psi}_1, \ldots, \boldsymbol{\psi}_p)$. The resulting $\mathbf{w}$ can then be re-fitted for each user $u$ via another pass of the algorithm in Section~\ref{sec:method-ssl} using only observations $\mathbf{r}_u$ from $u$'s individual surrogate function $r_u(\mathbf{x})$. This approach however does not optimize for how fast $\mathbf{w}$ can be adapted (on average) for a random user $u$ with limited data. To account for this, we instead minimize the following \emph{post-update, personalized} loss function over $q$ users -- see its intuition below Eq.~\eqref{eq:m3},
\begin{eqnarray}
\mathbf{L}(\mathbf{w}) &=& \frac{1}{q}\sum_{u=1}^q\Big[\ell_u(\boldsymbol{\psi}, \kappa_u(\mathbf{w}))\Big] \label{eq:m2}
\end{eqnarray}
where $\ell(\boldsymbol{\psi}, \mathbf{w})$ is defined in Eq.~\eqref{eq:m1} above and $\ell_u(\boldsymbol{\psi}, \kappa_u(\mathbf{w}))$ is identical in form to $\ell(\boldsymbol{\psi})$ except for the fact that it is parameterized by $\mathbf{w}_u = \kappa_u(\mathbf{w})$ (instead of $\mathbf{w}$) and computed based on local observations $\mathbf{r}_u$ (instead of $\mathbf{r}$). Here, $\kappa_u(\mathbf{w})$ denotes a personalization procedure that minimizes $\ell_u(\boldsymbol{\psi}, \mathbf{w})$ for each user $u$, which can be represented in the following form,
\begin{eqnarray}
\hspace{-12mm}\kappa_u(\mathbf{w}) \ =\ \kappa_u^{(t)}(\mathbf{w}) &\triangleq&\kappa_u^{(t-1)}(\mathbf{w}) \nonumber\\
&+& \omega \cdot  \nabla_{\mathbf{w}}\ell\left(\kappa_u^{(t-1)}(\mathbf{w})\right) \label{eq:m3}
\end{eqnarray}
and $\kappa_u^{(0)}(\mathbf{w}) = \mathbf{w}$. Here, we drop $\boldsymbol{\psi}$ from the argument of $\ell_u(\boldsymbol{\psi},\mathbf{w})$ since it is clear from context and is fixed. This encompasses the $t$-step gradient update procedure that aims to numerically minimize $\ell_u(\mathbf{w})$ with $\mathbf{w}$ being the initializer and $\omega$ denote the learning rate. Intuitively, minimizing Eq.~\eqref{eq:m2} means finding a vantage point $\mathbf{w}$ that are most effective for personalization. That is, starting at $\mathbf{w}$, the local update procedure $\kappa_u(\mathbf{w})$ can arrive at an effective parameter configuration that reduces the local loss $\ell_u(\mathbf{w})$ the most. This generic form, however, poses a challenge since the gradient $\nabla_{\mathbf{w}}\mathbf{L}$ might not be tractable since $\kappa_u(\mathbf{w})$ might not exist in closed-form. 

{\bf Key Idea:} To address this, note that the vector-value function $\kappa_u(\mathbf{w})$ can be represented as $\kappa_u(\mathbf{w}) = [\kappa_u^1(\mathbf{w}) \ldots \kappa_u^{p + 1}(\mathbf{w})]$ where we have $\mathrm{dim}(\mathbf{w}) = \mathrm{dim}(\mathbf{h}) + \mathrm{dim}(b) = p + 1$. We can then approximate each component with a $2^{\mathrm{nd}}$-order Taylor expansion around $\mathbf{0}$ and show that under such expansion, $\nabla_{\mathbf{w}}\mathbf{L}$ can be computed, which allows $\mathbf{L}$ to be minimized via Lemma~\ref{lem:0}.

\begin{lemma}
\label{lem:0}
Assuming $\kappa_u^i(\mathbf{w})$ is twice-differentiable at $\mathbf{w} = \mathbf{0}$, the approximation of $\kappa_u^i(\mathbf{w})$ with its $2^{\mathrm{nd}}$-order Taylor expansion around $\mathbf{w} = \mathbf{0}$ induces $\nabla_{\mathbf{w}}\mathbf{L}(\mathbf{w}_+) = $
\begin{eqnarray}
\hspace{-15mm}\frac{1}{q}\sum_{u=1}^{q}\Bigg(\Big[\mathbf{D}_{\mathbf{w}}\kappa_u(\mathbf{w}_+)\Big]\nabla_{\mathbf{w}}\ell_u\Big(\kappa_u(\mathbf{w}_+)\Big)\Bigg) \label{eq:m4}
\end{eqnarray}
with $\mathbf{D}_{\mathbf{w}}\kappa_u(\mathbf{w}_+) \triangleq \Big[\nabla^\top_{\mathbf{w}}\kappa_u^1(\mathbf{w}_+); \ldots; \nabla^\top_{\mathbf{w}}\kappa_u^{p + 1}(\mathbf{w}_+)\Big]$ whose rows are approximated via
\begin{eqnarray}
\hspace{-13mm}\nabla_{\mathbf{w}}\kappa_u^i\Big(\mathbf{w}_+\Big) \hspace{-2mm}&\simeq&\hspace{-2mm} \nabla_{\mathbf{w}}\kappa_u^i\Big(\mathbf{0}\Big) +   \left[\nabla^2_{\mathbf{w}}\kappa_u^i\Big(\mathbf{0}\Big)\right]\mathbf{w}_+ \label{eq:m5}
\end{eqnarray}
\end{lemma}

Lemma~\ref{lem:0} implies that if $\nabla_{\mathbf{w}}\ell(\mathbf{w}_+)$ and  $\nabla_{\mathbf{w}}\ell_u(\mathbf{w}_+)$ are tractable; and $\nabla_{\mathbf{w}}\kappa_u^i(\mathbf{w}_+)$ and $\nabla^2_{\mathbf{w}}\kappa_u^i(\mathbf{w}_+)$ are tractable at $\mathbf{w} = \mathbf{0}$ then the gradient of $\mathbf{L}(\mathbf{w})$ is also approximately tractable at any $\mathbf{w}_+$, thus mitigating the lack of a closed-form expression for $\kappa_u(\mathbf{w})$. Here, the tractability of $\nabla_{\mathbf{w}}\ell(\mathbf{w}_+)$ and $\nabla_{\mathbf{w}}\ell_u(\mathbf{w}_+)$ is evident from their analytic form in Eq.~\eqref{eq:m1} while the tractability of $\nabla_{\mathbf{w}}\kappa_u^i(\mathbf{w}_+)$ and $\nabla^2_{\mathbf{w}}\kappa_u^i(\mathbf{w}_+)$ at $\mathbf{w}_+ = \mathbf{0}$ along with the rest of the proof is deferred to Appendix~\ref{app:g}. 

Note that for simple choices of $\kappa_u(\mathbf{w})$, $\mathbf{D}_{\mathbf{w}}\kappa_u(\mathbf{w}_+)$ can be computed analytically to bypass this approximation. For instance, in our experiment, we choose $\kappa_u(\mathbf{w}) = \mathbf{w} - \omega\nabla_{\mathbf{w}}\ell_u(\mathbf{w})$. It then follows that $\mathbf{D}_{\mathbf{w}}\kappa_u(\mathbf{w}_+) = \mathbf{I} - \omega\nabla^2_{\mathbf{w}}\ell_u(\mathbf{w}_+)$ which is exact and tractable. This leads to a simpler expression for Eq.~\eqref{eq:m4}, which mimics the update equation in meta learning \citep{Finn17}.

\section{Theoretical Analysis}
\label{sec:theory}
This section analyzes the effectiveness of the proposed self-supervised metric learning algorithm (Section~\ref{sec:method-ssl}) from a theoretical lens, which aims to shed insights on when and how the induced metric approximates accurately. In essence, our main results, Theorems~\ref{theo:1} and~\ref{theo:2}, show that under reasonable assumptions, the induced metric $\mathbf{D}(\mathbf{x},\mathbf{x}')$ of our algorithm is arbitrarily close to the true (unknown) metric $\mathbf{D}_\ast(\mathbf{x},\mathbf{x}')$ with high probability if it can observe a sufficiently large number $n$ of observations from the surrogate training feedback $r(\mathbf{x})$. Our assumptions are first stated below.

{\bf Assumptions.} Let $\mathbf{D}_\ast(\mathbf{x}, \mathbf{x}')$ denote the (unknown) true metric function and $k_\ast(\mathbf{x},\mathbf{x}') = \mathrm{exp}(-0.5\cdot\mathbf{D}_\ast(\mathbf{x}, \mathbf{x}'))$ denote the oracle kernel parameterized by the true metric. Then:

{\bf A1.} There exists a constant value $d > 0$ for which $\displaystyle d\cdot \mathrm{inf}\  k_\ast(\mathbf{x},\mathbf{x}') \geq \lambda_{\mathrm{max}}(\mathbf{K}_\ast)$ where $\lambda_{\max}(\mathbf{K}_\ast)$ denote the largest eigenvalue of the Gram matrix induced by $k_\ast$ on $\{\mathbf{x}_1, \mathbf{x}_2, \ldots, \mathbf{x}_n\}$.

{\bf A2.} Let $\widehat{\mathbf{r}} = [\widehat{r}(\mathbf{x}_1), \ldots, \widehat{r}(\mathbf{x}_n)]^\top$ denote the GP prediction of $\mathbf{r} = [r(\mathbf{x}_1), \ldots, r(\mathbf{x}_n)]^\top$ using the learned kernel function $k(\mathbf{x}, \mathbf{x}')$ -- see Eq.~\eqref{eq:6}\footnote{As it is clear from context, we drop the parameterization notation $\boldsymbol{\psi}$ from this point onward for simplicity.}, there exists a non-negative constant $\alpha$ s.t. $\underset{(\mathbf{x},\mathbf{x}')}{\mathrm{sup}}\ |k(\mathbf{x}, \mathbf{x}') - k_\ast(\mathbf{x}, \mathbf{x}')| \leq$
\begin{eqnarray}
\hspace{-25.5mm}1 &-& \alpha\ \cdot\ \Bigg(\Big(\mathbf{r} -  \widehat{\mathbf{r}}\Big)^\top\mathbf{A}\Big(\mathbf{r} -  \widehat{\mathbf{r}}\Big)\Bigg)^{-1} \label{eq:9}
\end{eqnarray}
where $\displaystyle\mathbf{A} = \frac{1}{n}\left(\mathbf{K} + \sigma^2\mathbf{I}\right)^2$ and $\mathbf{K}$ is the Gram matrix induced by $k(\mathbf{x},\mathbf{x}')$ on $\{\mathbf{x}_1, \mathbf{x}_2, \ldots, \mathbf{x}_n\}$.

{\bf A3.} Let $\mathbf{r} \ =\ [r(\mathbf{x}_1), r(\mathbf{x}_2), \ldots, r(\mathbf{x}_n)]^\top$ and $\mathbf{K}_\ast$ be defined as above. We assume $\mathbf{r} \sim \mathbb{N}(0, \mathbf{K}_\ast)$. This is key to establish our main results in Theorems 1 and 2.

{\bf Remark.} Here, assumptions $\mathbf{A1}$ and $\mathbf{A2}$ state that the oracle kernel is bounded from below ({\bf A1}) and the discrepancies between the approximate and oracle kernel are bounded above with a ceiling no more than $1$ ({\bf A2}). These are reasonable assumptions which can be realized in most cases given that by construction, the range of values for both kernel functions is between $0$ and $1$. Last, while $\mathbf{A3}$ imposes a stronger assumption on the statistical relationship between the surrogate training feedback $r(\mathbf{x})$ and oracle kernel $\mathbf{K}_\ast$, this is also not unreasonable given that in many situations, there also exists many feature signals that are both normally distributed and are directly related to the similarities across input instances, e.g. measurements of height/weight among people.

Under these assumptions, we can now state our key results which provably demonstrates how well the self-supervised learning metric can approximated the true metric, and under what conditions. Our strategy to address these questions are first described below. 

{\bf Analysis Strategy.} First, we aim to establish that if the maximum multiplicative error in approximating the oracle kernel (parameterized by the true metric) with the induced kernel (parameterized by the metric learned by our algorithm) can be made arbitrarily small, the same can also be said about the discrepancies between the true and learned metric (see Lemmas~\ref{lem:1} and~\ref{lem:2}). 

Then, we further establish that while the maximum multiplicative error between kernels is not always small with $100\%$ certainty, the probability that it is large is vanishingly small as the size of the dataset increases (Theorem~\ref{theo:1}). This implies the results of Lemmas~\ref{lem:1} and ~\ref{lem:2} can be invoked with high chance, guaranteeing that the metric discrepancies can be made vanishingly small with high probability. This is formalized in Theorem~\ref{theo:2}, which also details the least amount of data necessary for such event to happen.\vspace{2mm}

{\bf Formal Results.} To begin our technical analysis, we start with Lemma~\ref{lem:1} below which shows that if the ratio between the true and approximate kernel values, $k_\ast(\mathbf{x}, \mathbf{x}')$ and $k(\mathbf{x},\mathbf{x}')$, can be made arbitrarily close at $(\mathbf{x}, \mathbf{x}')$ then the approximated distance metric is also arbitrarily close at $(\mathbf{x}, \mathbf{x}')$.

\begin{lemma} 
\label{lem:1}
Suppose $\displaystyle (1 \ -\ \epsilon)\cdot k_\ast(\mathbf{x},\mathbf{x}') \ \ \leq\ \  k(\mathbf{x},\mathbf{x}')\ \ \leq\ \ (1 \ +\  \epsilon) \cdot k_\ast(\mathbf{x},\mathbf{x}')$ for $\epsilon \in (0, 1)$ then
\begin{eqnarray}
\hspace{-15mm}\Big|\mathbf{D}(\mathbf{x}, \mathbf{x}') - \mathbf{D}_\ast(\mathbf{x},\mathbf{x}')\Big| &\leq& 2\log\left(\frac{1}{1 - \epsilon}\right) \ . \label{eq:10}
\end{eqnarray}
\end{lemma}

The result of Lemma~\ref{lem:1}, which is formally proved in Appendix~\ref{app:a}, suggests a direct strategy to guarantee the metric approximation is close to zero simultaneously for all pair $(\mathbf{x},\mathbf{x}')$, as formalized in Lemma~\ref{lem:2}.

\begin{lemma}
\label{lem:2}
Suppose $\displaystyle \underset{(\mathbf{x},\mathbf{x}')}{\mathrm{sup}}\ \left|\frac{k(\mathbf{x},\mathbf{x}') - k_\ast(\mathbf{x}, \mathbf{x}')}{k_\ast(\mathbf{x},\mathbf{x}')}\right| \ \leq\ \epsilon$ for $\epsilon \in (0, 1)$ then
\begin{eqnarray}
\hspace{-8mm}\forall (\mathbf{x}, \mathbf{x}'):\Big|\mathbf{D}(\mathbf{x}, \mathbf{x}') - \mathbf{D}_\ast(\mathbf{x},\mathbf{x}')\Big| \hspace{-2mm}&\leq&\hspace{-2mm} 2\log\left(\frac{1}{1 - \epsilon}\right) \label{eq:13}
\end{eqnarray}
\end{lemma}

Enforcing the premise of Lemma~\ref{lem:2} -- see its proof in Appendix~\ref{app:b} -- however, is not always possible as it depends on the randomness in which we obtain our observations of the surrogate training feedback $r(\mathbf{x})$. This raises the following questions:\vspace{2mm}

{\bf How likely this happens and how many observations are sufficient to guarantee that such premise would happen with high chance?} \vspace{2mm}

These are addressed in Theorems~\ref{theo:1} and ~\ref{theo:2} below.\vspace{2mm}

\begin{theorem} 
\label{theo:1}
Let $\displaystyle g(\tau) \triangleq \log(\tau) + (1/\tau) - 1$ and $\displaystyle c_{\epsilon} \triangleq \epsilon\lambda_{\max}(\mathbf{K}_\ast)/d$, we have
\begin{eqnarray}
\hspace{-15mm}&\ &\mathrm{Pr}\Bigg(\underset{(\mathbf{x},\mathbf{x}')}{\mathrm{sup}}\left|\frac{k(\mathbf{x},\mathbf{x}') - k_\ast(\mathbf{x},\mathbf{x}')}{k_\ast(\mathbf{x},\mathbf{x}')}\right| \geq\epsilon\Bigg)\nonumber\\
\hspace{-15mm}&\leq& \mathrm{exp}\left(-\frac{1}{2}n \cdot g\left(\frac{\sigma^4}{\alpha}\Big(1 - c_{\epsilon}\Big)\lambda_{\max}\Big(\mathbf{K}_\ast\Big)\right)\right) \label{eq:14}
\end{eqnarray}
where the constants $d$ and $\alpha$ are defined in {\bf A1} and {\bf A2} above respectively. The detailed proof of Theorem~\ref{theo:1} is deferred to Appendix~\ref{app:c}. In addition, we note that $g(\tau)$ is non-negative so the RHS of Eq.~\eqref{eq:14} is no greater than $1$, ensuring that the bound is not vacuous.
\end{theorem}

Theorem~\ref{theo:1} therefore establishes that the premise of Lemma~\ref{lem:2} will happen with a high probability with a sufficiently large value of $n$, thus asserting its implication of Lemma~\ref{lem:2} with high chance. This guarantees the discrepancy between the learned and true metric at any pair $(\mathbf{x}, \mathbf{x}')$ can be made arbitrarily small with a large value of $n$. This is formalized below.\vspace{2mm}

\begin{theorem}
\label{theo:2}
Let $\displaystyle g(\tau) = \log(\tau) + (1/\tau) - 1$ and $\displaystyle g_{\epsilon} = g\left(\frac{\sigma^4}{\alpha}\left(1 - \lambda_{\max}\Big(\mathbf{K}_\ast\Big)\frac{\epsilon}{d}\right)\lambda_{\max}\Big(\mathbf{K}_\ast\Big)\right)$. Then,
\begin{eqnarray}
\mathrm{Pr}\Bigg(\underset{(\mathbf{x},\mathbf{x}')}{\mathrm{sup}}\ \Big|\mathbf{D}(\mathbf{x},\mathbf{x}') - \mathbf{D}_\ast(\mathbf{x},\mathbf{x}')\Big| &\leq& 2\log\left(\frac{1}{1 - \epsilon}\right)\Bigg)\nonumber\\ &\geq& 1 \ -\ \delta \label{eq:15}
\end{eqnarray}
when $n \geq\ \frac{2}{g_{\epsilon}}\log\frac{1}{\delta}$ and $\delta \in (0, 1)$ is an arbitrarily small confidence parameter. Theorem~\ref{theo:2} can be derived from Theorem~\ref{theo:1} by setting the RHS of Eq.~\eqref{eq:14} to an arbitrarily small value $\delta$, solving for $n$ and following up with direct application of Lemma~\ref{lem:2}. Its proof is deferred to Appendix~\ref{app:d}.
\end{theorem}
Theorem~\ref{theo:2} thus concludes our analysis with the following take-home message: Under reasonable assumptions in {\bf A1, A2} and {\bf A3}, the metric induced by our self-supervised learning algorithm is vanishingly close to the true metric with arbitrarily high probability provided that we have access to a sufficiently large dataset of the surrogate training feedback $r(\mathbf{x})$. The statistical relation between this surrogate feedback and the true metric as stated in {\bf A3} is key to establish this result.

\section{Experiment}
\label{sec:exp}
We evaluate our proposed self-supervised and personalized metric learning algorithms on the MovieLens \citep{harper2015movielens} and Yelp review dataset\footnote{https://www.yelp.com/dataset/download}. A short description of the datasets is provided below.

{\bf MovieLens Dataset.} The dataset comprises $26$K+ items which were interacted with by $138$K+ users. Each interaction is a triplet of user, item and timestamp which is measured in seconds with respect to a certain point of origin in $1970$. There are about $20$M of such interactions and in addition, the dataset also provides multiple channels of meta data per item in various formats such as categorical (genre), numerical (rating) and text (plot and title). Here, representations of categorical and text features are multi-hot and pre-trained BERT \citep{devlin2018bert} embedding vectors, respectively.

{\bf Yelp Review Dataset.} The dataset comprises $8$M+ reviews given to businesses by customers. Here, we treat the businesses as items and customers as users. There are approximately $160$K+ businesses (items) and about $2$M+ customers (users). For each business, we have meta data regarding its averaged rating and business categories. The latter of which is represented as a multi-hot vector ranging over $1300$+ categories (e.g., Burgers, Mexican and Gastropubs). 

Both datasets were pre-processed using the same procedure as described in Section~\ref{app:f1}. In what follows, our experiments aim to address the following key questions:


{\bf Q1.} Does the induced metric via self-supervised learning (Section~\ref{sec:method-ssl}) improve over the vanilla metric induced from optimizing a Siamese network on noisy annotations of similar pairs of items?

{\bf Q2.} Does such SSL induced metric can be further personalized (Section~\ref{sec:method-personalize}) to fit better with a user's personal notion of similarity, which often varies substantially across different users?

\begin{figure*}[t]
\begin{tabular}{ccc}
\centering
\vspace{-3mm}
\hspace{-4mm}\includegraphics[width=0.33\linewidth]{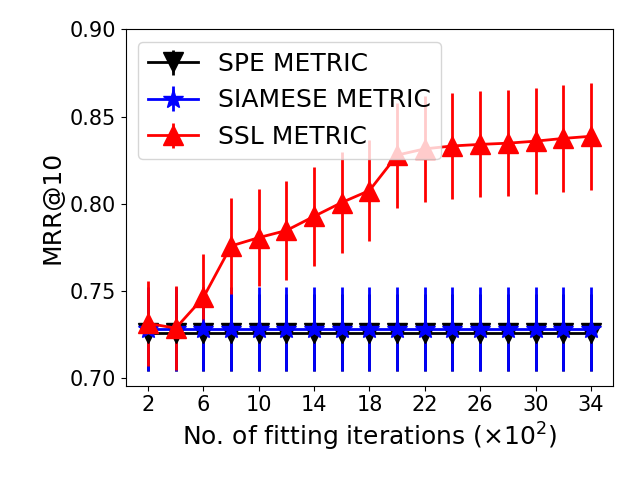} & \hspace{-3mm}\includegraphics[width=0.33\linewidth]{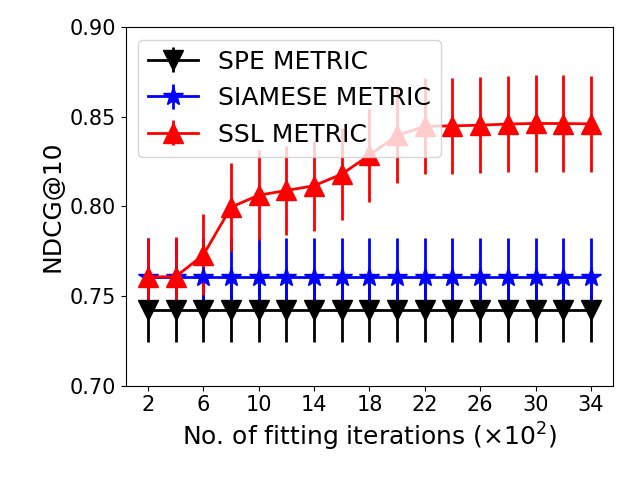}&
\hspace{-3mm}\includegraphics[width=0.33\linewidth]{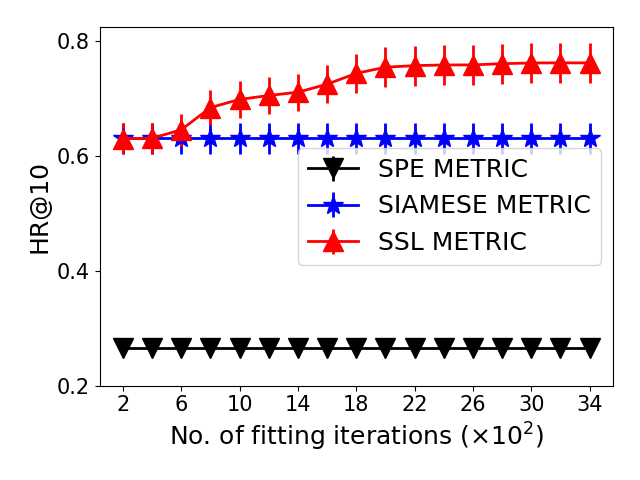}
\end{tabular}
\caption{Plots of the averaged MRR, NDCG and HR measurements (over $5$ independent SSL runs) of our SSL metric when it was used to produce a top-$10$ recommendation list for each unseen test item from the MovieLens dataset. The corresponding measurements of SPE and Siamese baselines were also included as references. These measurements however do not improve with more iterations as SPE and Siamese are non-SSL.}
\label{fig:ssl_exp_movielen_maintext}
\end{figure*}

\begin{figure*}[h!]
\begin{tabular}{ccc}
\centering
\vspace{-3mm}
\hspace{-4mm}\includegraphics[width=0.33\linewidth]{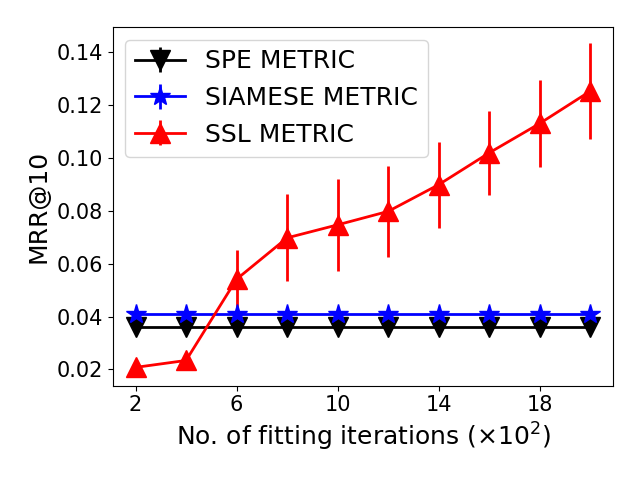} & \hspace{-3mm}\includegraphics[width=0.33\linewidth]{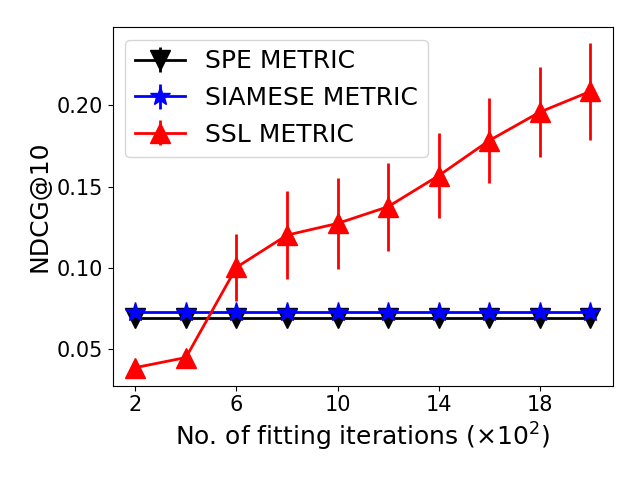}&
\hspace{-3mm}\includegraphics[width=0.33\linewidth]{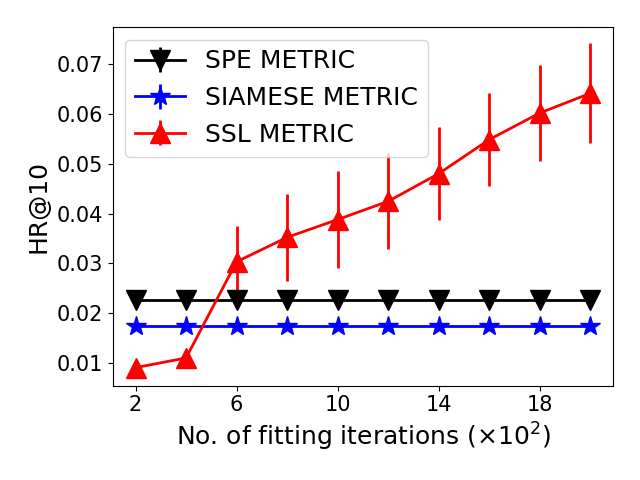}
\end{tabular}
\caption{Plots of the averaged MRR, NDCG and HR measurements (over $5$ independent SSL runs) of our SSL metric evaluated on YELP dataset. The performance of SPE and Siamese baselines were also included as references. These measurements however do not improve with more iterations as SPE and Siamese are non-SSL. Similar to our experiments with the MovieLens dataset, the MRR, NDCG and HR measurements are computed with respect to top-$10$ recommendation lists produced by the participating baselines on the same unseen test set.}
\label{fig:ssl_exp_yelp}
\end{figure*}

{\bf Evaluation.} Once learned, the item-to-item metric is used to rank the items in the list of candidates (i.e., the entire item catalogue) in the decreasing order of their similarities to a test item. For each test item, the quality of the resulting ranked list can then be assessed via standard ranking measurements such as mean reciprocal rank (MRR), hit rate (HR) and normalized discounted cumulative gain (NDCG). See below. 

{\bf Measurement Description.} To evaluate the efficiency of an item metric $\mathbf{D}$ on a specific item $\mathbf{x}$, we use it to compute the distance between $\mathbf{x}$ and every other item $\mathbf{x}'$ in the catalogue. The top $k = 10$ closest items $\{\mathbf{x}'_1, \ldots, \mathbf{x}'_{k}\}$ based on their computed distances to $\mathbf{x}$ are then extracted. Let $\mathbf{G}(\mathbf{x})$ denote the set of similar items to $\mathbf{x}$, the following quality measurements are computed to assess $\mathbf{D}$:

{\bf HR@K.} The HR@K (or \emph{hit rate} at $k$) measurement of $\mathbf{D}$ at test item $\mathbf{x}$ is 
\begin{eqnarray}
\hspace{-16mm}\mathbf{HR}_k\Big(\{\mathbf{x}'_i\}_{i=1}^k; \mathbf{x}\Big) \hspace{-2mm}&\triangleq&\hspace{-2mm} k^{-1}\sum_{i=1}^{k} \mathbb{I}\Big(\mathbf{x}'_i \in \mathbf{G}(\mathbf{x})\Big)\nonumber
\end{eqnarray}
where $\{\mathbf{x}'_1, \ldots, \mathbf{x}'_k\}$ represent the top $k$ items suggested by the recommender, in decreasing order of relevance to the test item $\mathbf{x}$. The average HR@K is computed by averaging over items in a test set.

{\bf MRR@K.} The MRR@K (or \emph{mean reciprocal rank} at $k$) measurement of $\mathbf{D}$ at $\mathbf{x}$ is 
\begin{eqnarray}
\mathbf{MRR}_k\Big(\{\mathbf{x}'_i\}_{i=1}^k; \mathbf{x}\Big) \hspace{-3mm}&\triangleq&\hspace{-3mm} \Big(\argmin_{i=1}^k \{\mathbf{x}'_i: \mathbf{x}'_i \in \mathbf{G}(\mathbf{x})\}\Big)^{\hspace{-1mm}-1}\nonumber
\end{eqnarray}
or $0$ if none of the recommended items is in $\mathbf{G}(\mathbf{x})$. The average MRR@K is then computed by averaging over items in a test set.

{\bf NDCG@K.} First, the DCG@K -- \emph{discounted cumulative gain} at $k$ -- measurement of $\mathbf{D}$ at $\mathbf{x}$ is
\begin{eqnarray}
\hspace{-1.5mm}\mathbf{NDCG}_k\Big(\{\mathbf{x}'_i\}_{i=1}^k; \mathbf{x}\Big) \hspace{-2mm}&\triangleq&\hspace{-2mm} \mathbb{I}\Big(\mathbf{x}'_1 \in \mathbf{G}(\mathbf{x})\Big) \nonumber\\ \hspace{-2mm}&+&\hspace{-2mm} \sum_{i=2}^k \mathbb{I}\Big(\mathbf{x}'_i \in \mathbf{G}(\mathbf{x})\Big) \cdot \log_2^{-1}\left[i\right] \nonumber
\end{eqnarray}
The NDCG@K -- \emph{normalized discounted cumulative gain} at $K$ -- is then obtained by dividing DCG@K to the maximum achievable DCG@K by a permutation of items in the recommendation list. The average NDCG@K is computed by averaging over the test set.

The overall assessment of the item-to-item metric is then computed by averaging the above measurements over the set of test items. Here, the test items are those whose $1$st interaction happens after a timestamp (hence, not visible to the learning algorithms) which was set so that the test set comprises about $5\%$ of the entire item catalogue. 

To set the ground-truth for item-to-item recommendation, we deem two items similar if they were interacted with by the same users within a time horizon. Otherwise, they are deemed dissimilar. Here, we also note that unlike user-item truths, item-item truths acquired in this fashion are undeniably noisy so during training, we further make use of a downstream prediction task that (presumably) correlates well with the oracle item-item truths. To our intuition, rating prediction in the case of movies does appear to correlate with item similarities, which explains for the improved performance of our SSL method as reported in Section~\ref{sec:exp-ssl} below.

All experiments were run on a computing server with a Tesla V100 GPU with 16GB RAM. For more information regarding our experiment setup and data pre-processing, please refer to Appendix~\ref{app:f}.

\begin{figure*}[h!]
\begin{tabular}{ccc}
\centering
\vspace{-3mm}
\hspace{-4mm}\includegraphics[width=0.33\linewidth]{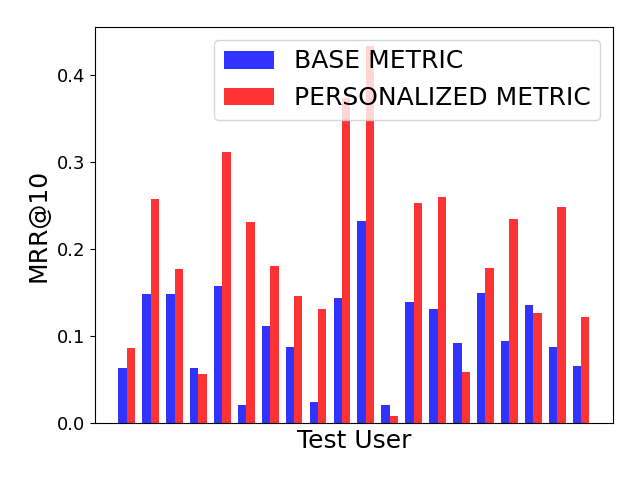} & \hspace{-3mm}\includegraphics[width=0.33\linewidth]{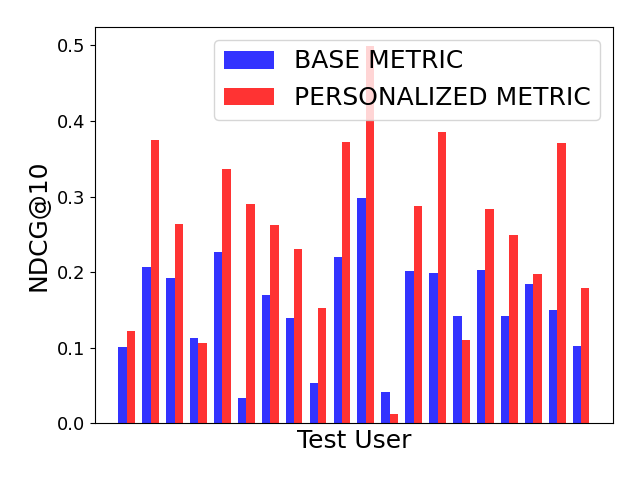}&
\hspace{-3mm}\includegraphics[width=0.33\linewidth]{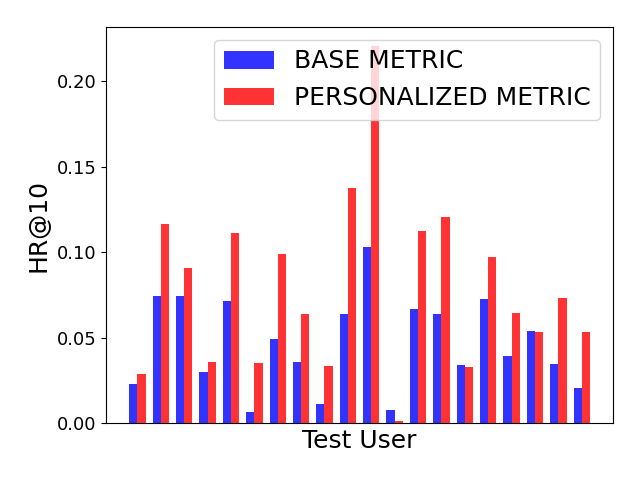}
\end{tabular}
\caption{Bar charts of the individual MRR, NDCG and HR measurements of a set of $20$ randomly sampled users with fewer than $200$ interactions with items in the catalogue. Red (blue) columns reflect the measurement (MRR, NDCG and HR) of the personalized (base) SSL metric on those users.}\vspace{-4mm}
\label{fig:personalize_exp_movielen_maintext}
\end{figure*}

\begin{figure*}[h!]
\begin{tabular}{ccc}
\centering
\vspace{-3mm}
\hspace{-4mm}\includegraphics[width=0.33\linewidth]{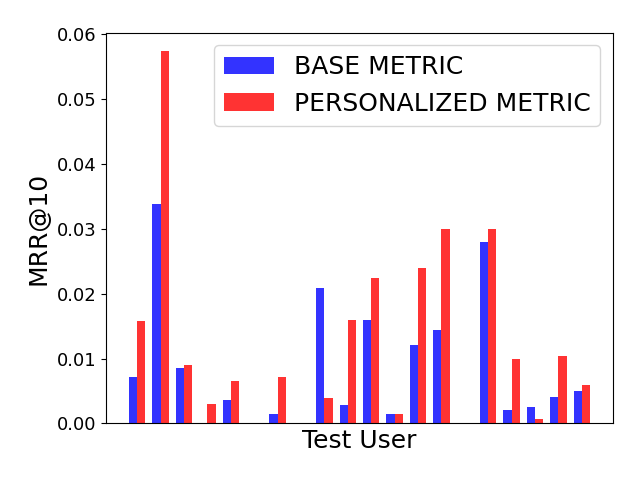} & \hspace{-3mm}\includegraphics[width=0.33\linewidth]{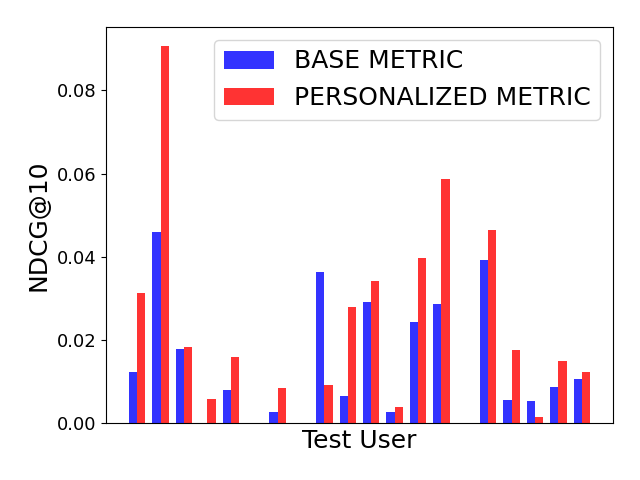}&
\hspace{-3mm}\includegraphics[width=0.33\linewidth]{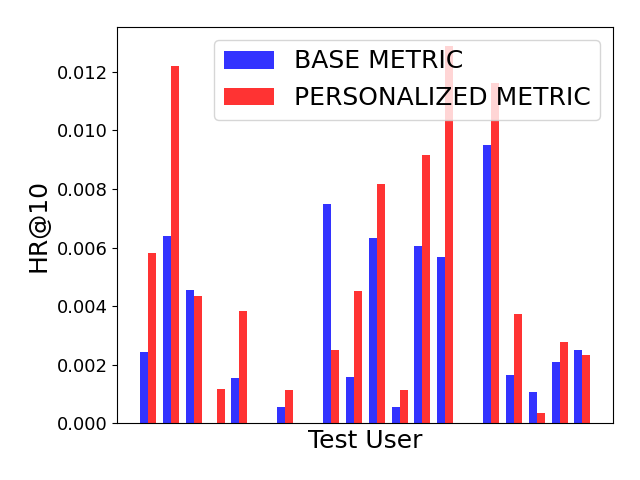}
\end{tabular}
\caption{Bar charts of the individual MRR, NDCG and HR measurements of our SSL metric on a set of randomly sampled users with fewer than $200$ interactions with items in the catalogue. Red (blue) columns reflect the quality measurement (MRR, NDCG and HR) of the personalized (base) SSL metric on those users.}
\label{fig:personalize_exp_yelp}
\end{figure*}

\subsection{Self-Supervised Metric Learning}
\label{sec:exp-ssl}
To answer {\bf Q1}, we evaluate the performance of the item-to-item metrics generated by (1) optimizing the vanilla Siamese ensemble (SIAMESE) combining the meta information channels (including \emph{ratings}) of the items; (2) optimizing the more recently proposed SPE method \citep{Li2019} using \emph{ratings} as side information and other channels as content; and (3) optimizing the GP with Siamese kernel (SSL), which is initialized with the Siamese ensemble generated in (1), using averaged \emph{ratings} of the items. The results were averaged over $5$ independent runs and reported in Figure~\ref{fig:ssl_exp_movielen_maintext} and Figure~\ref{fig:ssl_exp_yelp}. 

It is observed from both Figure~\ref{fig:ssl_exp_movielen_maintext} and Figure~\ref{fig:ssl_exp_yelp} that our SSL metric becomes increasingly better and outperforms both the semi-parametric embedding (SPE) and SIAMESE baselines significantly (across all measurements) after $3000$ iterations. This provides strong evidence to support our intuition earlier (Section~\ref{sec:intro}) that as we implicitly express the correlation between items in terms of the metric representation that parameterizes a GP model, a well-fitted GP would induce a well-shaped metric that preserves the averaged similarity geometry of items, as suggested in Section~\ref{sec:theory}.\vspace{-2mm}


\begin{figure*}[t]
\begin{tabular}{ccc}
\centering
\vspace{-3mm}
\hspace{-4mm}\includegraphics[width=0.33\linewidth]{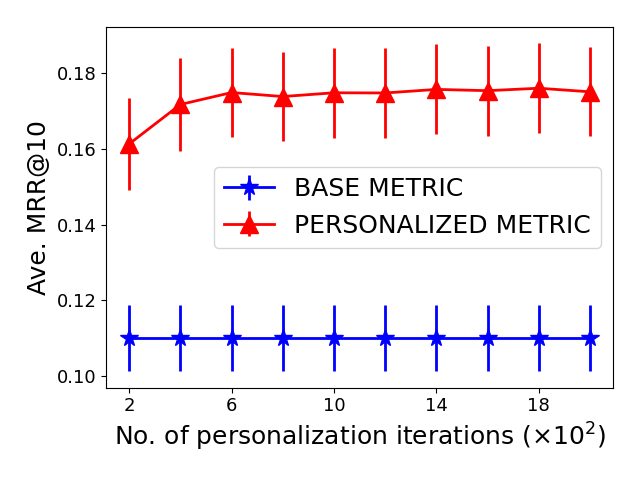} & \hspace{-3mm}\includegraphics[width=0.33\linewidth]{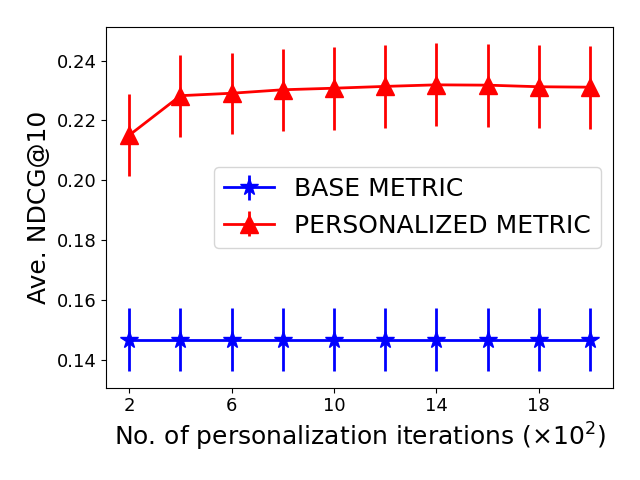}&
\hspace{-3mm}\includegraphics[width=0.33\linewidth]{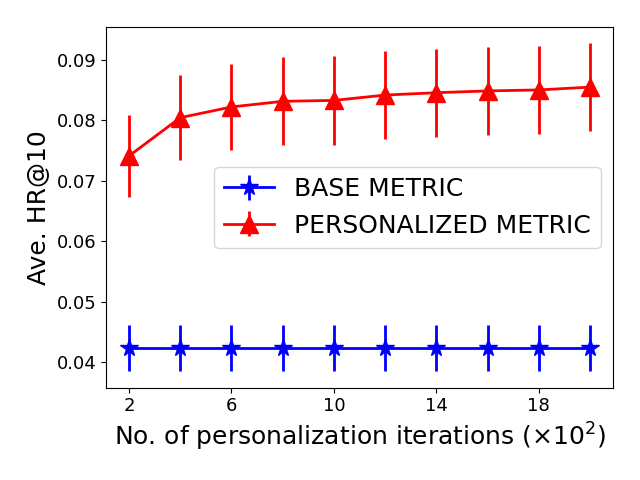}
\end{tabular}
\caption{Plots of the MRR, NDCG and HR measurements of our personalized metric over $2000$ personalization iterations on the MovieLens dataset. Our personalized metric was evaluated separately on each user and the plotted results were averaged over $20$ users. Here, its performance measurements on each individual user is computed using only the user's personal co-interaction data, rather than the co-interaction data over the entire population, which is stricter and expectedly so to evaluate personalized performance.}\vspace{-4mm}
\label{fig:personalize_average_exp_maintext}
\end{figure*}

\subsection{Personalized Metric Learning}
\label{sec:exp-personalize}

Though we obtain positive evidence in Section~\ref{sec:exp-ssl} that our SSL-induced metric improves significantly over all baselines, this is measured on the average over the entire user population rather than on each individual user. In the former context, as long as two items A and B are both interacted with by a user in the population, they are considered similar. But, in the (latter) context of a single user, A and B might not be considered similar if the user never interacts with both of them. In this case, we are interested in understanding how well our SSL metric would perform \emph{with} and \emph{without} personalization (see {\bf Q2}), especially on users with limited data. To shed light on this matter, we sample a subset of $20$ users with fewer than $200$ (but no fewer than $20$) item interactions. We then compute a personalizable metric based on the previously computed SSL-induced metric via minimizing Eq.~\eqref{eq:m2}. The personalizable metric is then tuned to fit each user's individual rating data using the SSL recipe in Section~\ref{sec:method-ssl}. 

The MRR, NDCG and HR measurements of the metric (before and after $2000$ personalization iterations on the MovieLens dataset) are reported in Figures~\ref{fig:personalize_exp_movielen_maintext} and~\ref{fig:personalize_average_exp_maintext}. Our observations are as follows: (1) on average (over $20$ users), the personalized metric consistently shows a sheer improvement over its SSL base in Figure~\ref{fig:personalize_average_exp_maintext}; (2) on an individual basis, the personalized MRR, NDCG and HR measurements improve significantly on $16/20$, $17/20$ and $17/20$ users, resulting in success rates of $80\%$, $85\%$ and $85\%$ (see Figure~\ref{fig:personalize_exp_movielen_maintext}). 

We also evaluate and report the individual personalized MRR, NDCG and HR measurements on the Yelp Review dataset for a sampled group of $20$ users in Figure~\ref{fig:personalize_exp_yelp}. In most individual cases, it can again be seen that the personalized metric also improves over the base metric (see Fig.~\ref{fig:personalize_exp_yelp}). These observations are therefore all consistent with our early observations on the MovieLens dataset. Note that, these are averaged measurements over the selected users. For each individual, the measurement is computed based \emph{only} on the corresponding user's personal co-interaction data, rather than on the co-interaction data over the entire population. Thus, this is a stricter performance measurement in comparison to that of Section~\ref{sec:exp-ssl}, and expectedly so to evaluate personalized performance.\vspace{-2mm}

\section{Conclusion}
\label{sec:conclusion}\vspace{-2mm}
This paper develops a self-supervised learning method for item-item metric distance in the context of a recommendation system where direct training examples are not readily available. The developed method embeds the metric representation as kernel parameters of a Gaussian process regression model, which is fitted on a surrogate prediction task whose training examples are more readily available. Our approach builds on the intuition that similar items should induce similar model prediction on the surrogate task, thus expressing and optimizing GP prediction in terms of item-item similarities can implicitly guide it towards capturing the right metric. We show that theoretically, our learning model can recover the right metric up to a certain error threshold with high probability. 

{\bf Societal Impact.} Though applications of our work to real data could result in ethical considerations, this is an unpredictable side-effect of our work. We use sanitized public datasets to evaluate its performance. No ethical considerations are raised.

\bibliography{aistats22}
\bibliographystyle{plainnat}

\clearpage
\appendix

\thispagestyle{empty}

\twocolumn[\makesupplementtitle]

\section{Additional Experimental Details}
\label{app:f}

\subsection{Experiment Setup}
\label{app:f1}
In our experiment setup, each dataset is pre-processed in the following forms that describe separately the item's behavior data (e.g., user-item interactions) and its content information (e.g., the meta data of an item that is user-independent). This includes:

{\bf Interaction Data.} This is a collection of tuples (\emph{user}, \emph{item}, \emph{rating}, \emph{timestamp}), each of which describes an event where a \emph{user} gives an \emph{item} a certain \emph{rating} at a certain time marked by \emph{timestamp}. Here, the \emph{timestamp} is in seconds counting from a certain origin in the past. For example, the time origin of the MovieLens dataset is at midnight (UTC) of January 1, 1970.

{\bf Meta Data.} This comprises multiple channels of different meta data. Each channel is a collection of pairs $(\mathbf{x}, \mathbf{v})$ where $\mathbf{x}$ is the item identification (e.g., a movie ID in a database) and $\mathbf{v}$ is either a scalar, multi-hot vector or a dense embedding vector representing a numerical feature, a categorical feature and an embedding of a text feature. For example, movie ratings would be represented as scalars, movie genres would be represented by multi-hot vectors, whose dimension is the total number of genres. Movie title and description are represented by $768$-dim embedding vectors generated by pre-trained BERT models \citep{devlin2018bert}.

{\bf Noisy Annotations of Similar (Dissimilar) Pairs.} To train a Siamese Net that captures the similarities between items, the standard method is to acquire (noisy) annotations of similar/dissimilar pairs of items. Here, for each user $\mathbf{u}$, we collect and sort the items $\{\mathbf{x}_1, \mathbf{x}_2, \ldots, \mathbf{x}_q\}$ that $\mathbf{u}$ has interacted with in the increasing order of time. For an randomly sampled item $\mathbf{x}_i \sim \{\mathbf{x}_1, \mathbf{x}_2, \ldots, \mathbf{x}_q\}$, we will sample (independently) $h$ items $\{\mathbf{x}^+_1, \mathbf{x}^+_2, \ldots, \mathbf{x}^+_{h}\}$ within a forward $\kappa$-step window $\{\mathbf{x}_{i + 1}, \mathbf{x}_{i + 2}, \ldots, \mathbf{x}_{i + \kappa}\}$. This forms $h$ positive (similar) pairs $\{(\mathbf{x}_i, \mathbf{x}^+_{\iota})\}_{\iota=1}^{h}$ per user. We also sample randomly (over the item catalogue) another $h$ items $\{\mathbf{x}^-_1, \ldots, \mathbf{x}^-_{h}\}$ to form another $h$ negative pairs $\{(\mathbf{x}_i, \mathbf{x}^-_{\iota})\}_{\iota=1}^{h}$. 

{\bf Training Siamese Net Baseline.} This results in a dataset $\{(\mathbf{x}_a,\mathbf{x}_b, y_{ab})\}$ where $y_{ab} = 0$ indicates $(\mathbf{x}_a,\mathbf{x}_b)$ is a pair of similar items and otherwise for $y_{ab} = 1$ (see Section~\ref{sec:siamese}). Here, the items $\mathbf{x}_a$ and $\mathbf{x}_b$ are represented as lists of meta vectors (one per channel), $\mathbf{x}_a = [\mathbf{v}_a^1, \ldots, \mathbf{v}_a^p]$ and $\mathbf{x}_b = [\mathbf{v}_b^1, \ldots, \mathbf{v}_b^p]$, respectively. This dataset will then be used to train a Siamese ensemble (see Section~\ref{sec:method-personalize}) which learns a separate individual Siamese distance $\mathbf{D}_i(\mathbf{v}^i_a, \mathbf{v}^i_b)$ per meta channel, and combines them via Eq.~\eqref{eq:7}. $\mathbf{D}_i(\mathbf{v}^i_a, \mathbf{v}^i_b)$ is computed via Eq.~\eqref{eq:6} which is expressed in the abstract form of a feature embedding tower $\mathbf{F}_{\mathbf{v}}$ and a scale matrix $\mathbf{\Lambda}$. Both of which are detailed next in Appendix~\ref{app:f2}.

{\bf Item Test Set.} The above annotation is restricted to items that appears before a certain time $\mathrm{T}$. An item is considered to appear before $\mathrm{T}$ if the earliest time it was interacted with by a user is before $\mathrm{T}$. In our experiment, $\mathrm{T}$ is selected such that the no. of test item is about $5\%$ of the item catalogue.

\subsection{Model Parameterization}
\label{app:f2}
This section elaborates further on the Siamese ensemble architecture that we mention in the main text. First, this refers to Eq.~\eqref{eq:7} which breaks down the overall item metric into individual metric across multiple meta channels. Second, each individual metric is characterized in Eq.~\eqref{eq:5} which concerns a feature embedding tower $\mathbf{F}(\mathbf{x}; \boldsymbol{\gamma})$ that maps the meta data vector of a single channel\footnote{Here, we abuse the notation $\mathbf{x}$ to represent a single-channel meta vector whereas previously, $\mathbf{x}$ was also used to denote a list of such single-channel meta vectors. Nonetheless, we believe this notation abuse does not impact the readability of the section since the context is clear.} into a metric space equipped with a Mahalonobis distance parameterized by a diagonal scale matrix $\boldsymbol{\Lambda}$. Here, we set $\boldsymbol{\Lambda} = \mathbf{I}$ seeing that its values on the diagonal can be absorbed into the parameterization $\boldsymbol{\gamma}$ of $\mathbf{F}(\mathbf{x}; \boldsymbol{\gamma})$, which is parameterized separately for each channel, as detailed below.

For each channel $m$ of meta data, the corresponding feature embedding tower $\mathbf{F}_m(\mathbf{x}; \boldsymbol{\gamma})$ is a feed-forward net comprising $3$ dense layers with $2 \times h$, $h$ and $h$ hidden units where $h = 50$ in our experiments. The layers are activated by a $\mathrm{ReLu}$ (RL), $\mathrm{sigmoid}$ and $\mathrm{tanh}$ functions in that order. That is, $\mathbf{F}_m(\mathrm{meta}_m(\mathbf{x}); \boldsymbol{\gamma}_m) \triangleq$
\begin{eqnarray}
\mathrm{tanh}\Bigg(\mathbf{W}^m_t\sigma\Big(\mathbf{W}^m_s\mathrm{RL}\Big(\mathbf{W}^m_o\mathrm{vec}(\mathbf{x}) \hspace{-1mm}+\hspace{-1mm}\mathbf{b}^m_o\Big) \hspace{-1mm}+\hspace{-1mm}\mathbf{b}^m_s\Big) + \mathbf{b}^m_t\Bigg) \nonumber
\end{eqnarray}
where $\mathbf{W}^m_o \in \mathbb{R}^{2h \times d}$, $\mathbf{W}^m_s \in \mathbb{R}^{h \times 2h}$ and $\mathbf{W}^m_t \in \mathbb{R}^{h \times h}$ are learnable affine transformation weights. Likewise, $\mathbf{b}^m_o \in \mathbb{R}^{2h}$, $\mathbf{b}^m_s \in \mathbf{R}^h$ and $\mathbf{b}^m_t \in \mathbf{R}^h$ are learnable bias vectors. Here, the $\mathrm{vec}$ or $\mathrm{Flatten}$ operator reshapes the input tensor $\mathbf{x}$ into a column vector of $d$ dimensions. The other $\mathrm{tanh}$, $\mathrm{sigmoid}$ and $\mathrm{ReLu}$ are applied point-wise to components of their input vectors or matrices. Thus, in short, we have $\boldsymbol{\gamma}_m = \{(\mathbf{W}^m_t, \mathbf{b}^m_t), (\mathbf{W}^m_s, \mathbf{b}^m_s), (\mathbf{W}^m_o, \mathbf{b}^m_o)\}$ as learnable parameters for $\mathbf{F}_m$.

Furthermore, in addition to the above parametric embedding of each meta channel, we also have a non-parametric embedding tower that maps from each item ID (i.e., an integer scalar) to a unique continuous $p$-dimensional vector. This ID embedding tower is non-parametric since its number of parameters is proportional to the number of items in the catalogue,
\begin{eqnarray}
\mathbf{F}_{\mathrm{ID}}\left(\mathrm{ID}\Big(\mathbf{x}\Big); \boldsymbol{\gamma}_{\mathrm{ID}}\right) \hspace{-3mm}&\triangleq&\hspace{-3mm} \mathrm{Flatten}\Bigg(\mathrm{Embedding}\Big(\mathrm{ID}(\mathbf{x}); \boldsymbol{\gamma}\Big)\Bigg) \nonumber
\end{eqnarray}
Here, $\boldsymbol{\gamma}_{\mathrm{ID}}$ comprises $p \times n$ learnable scalars where $n$ is the total number of items. The $\mathrm{Flatten}$ operator again reshapes the output into a $p$-dimension column vector. In our experiment, $p=30$. For more detail, our experimental code is also included in the supplement.

\section{Proof of Lemma~\ref{lem:0}}
\label{app:g}

{\bf Lemma 1.} Assuming $\kappa_u^i(\mathbf{w})$ is twice-differentiable at $\mathbf{w} = \mathbf{0}$, the approximation of $\kappa_u^i(\mathbf{w})$ with its $2^{\mathrm{nd}}$-order Taylor expansion around $\mathbf{w} = \mathbf{0}$ induces the following,
\begin{eqnarray}
\hspace{-4mm}\nabla_{\mathbf{w}}\mathbf{L}(\mathbf{w}_+) \hspace{-3mm}&=&\hspace{-3mm}  \frac{1}{q}\sum_{u=1}^{q}\Bigg(\Big[\mathbf{D}_{\mathbf{w}}\kappa_u(\mathbf{w}_+)\Big]\nabla_{\mathbf{w}}\ell_u(\kappa_u(\mathbf{w}_+))\Bigg) \nonumber\\\label{eq:g1}
\end{eqnarray}
with $\mathbf{D}_{\mathbf{w}}\kappa_u(\mathbf{w}_+) \triangleq \Big[\nabla^\top_{\mathbf{w}}\kappa_u^1(\mathbf{w}_+); \ldots; \nabla^\top_{\mathbf{w}}\kappa_u^{p + 1}(\mathbf{w}_+)\Big]$ whose rows are approximated via
\begin{eqnarray}
\hspace{-10mm}\nabla_{\mathbf{w}}\kappa_u^i\Big(\mathbf{w}_+\Big) \hspace{-2mm}&\simeq&\hspace{-2mm} \nabla_{\mathbf{w}}\kappa_u^i\Big(\mathbf{0}\Big) + \left[\nabla^2_{\mathbf{w}}\kappa_u^i\Big(\mathbf{0}\Big)\right]\ \mathbf{w}_+ \ .\label{eq:g2}
\end{eqnarray}

\begin{proof}
First, it is straight-forward to see that Eq.~\eqref{eq:g1} above can be derived by taking derivative on both sides of Eq.~\eqref{eq:m2} and the RHS of Eq.~\eqref{eq:g1} concerning the Jacobian $\mathbf{D}_{\mathbf{w}}\kappa_u(\mathbf{w}_+))$ is simply the result of the chain rule of differentiation. 

Thus, what remains to be proved is how one arrive at Eq.~\eqref{eq:g2} assuming the $2^{\mathrm{nd}}$-order Taylor expansion of $\kappa_u^i(\mathbf{w})$ exists around $\mathbf{w} = \mathbf{0}$ and can be used as a reasonable approximation. To see this, let us express the $2^{\mathrm{nd}}$-order Taylor of $\kappa_u^i(\mathbf{w})$ around $\mathbf{0}$ explicitly below:
\begin{eqnarray}
\kappa_u^i\Big(\mathbf{w}\Big) &\simeq& \kappa_u^i\Big(\mathbf{0}\Big) \ +\  \mathbf{w}^\top\left[\nabla_{\mathbf{w}}\kappa_u^i\Big(\mathbf{0}\Big)\right] \nonumber\\
&+& \frac{1}{2}\mathbf{w}^\top\left[\nabla^2_{\mathbf{w}}\kappa_u^i\Big(\mathbf{0}\Big)\right]\mathbf{w} \ .
\label{eq:g3}
\end{eqnarray}
Taking derivative with respect to $\mathbf{w}$ on both sides of Eq.~\eqref{eq:g3} and evaluating both at $\mathbf{w} = \mathbf{w}_+$ yield,
\begin{eqnarray}
\hspace{-13mm}\nabla_{\mathbf{w}}\kappa_u^i\Big(\mathbf{w}\Big) \hspace{-2mm}&\simeq&\hspace{-2mm} \nabla_{\mathbf{w}}\kappa_u^i\Big(\mathbf{0}\Big) + \left[\nabla^2_{\mathbf{w}}\kappa_u^i\Big(\mathbf{0}\Big)\right]\mathbf{w}_+ \ ,
\label{eq:g3}
\end{eqnarray}
which is the desired approximation above.
\end{proof}

\section{Proof of Lemma~\ref{lem:1}}
\label{app:a}

{\bf Lemma 2.}
Suppose $\displaystyle (1 \ -\ \epsilon) \cdot k_\ast(\mathbf{x},\mathbf{x}') \ \ \leq\ \  k(\mathbf{x},\mathbf{x}')\ \ \leq\ \ (1 \ +\  \epsilon) \cdot k_\ast(\mathbf{x},\mathbf{x}')$ for $\epsilon \in (0, 1)$ then
\begin{eqnarray}
\hspace{-11mm}\Big|\mathbf{D}(\mathbf{x}, \mathbf{x}') - \mathbf{D}_\ast(\mathbf{x},\mathbf{x}')\Big| \ \ &\leq&\ \  2\log\left(\frac{1}{1 - \epsilon}\right) \ . \label{eq:a1}
\end{eqnarray}

\begin{proof}
From the above premises, we have
\begin{eqnarray}
\hspace{-20mm}1 + \epsilon &\geq& k(\mathbf{x},\mathbf{x}')/k_\ast(\mathbf{x},\mathbf{x}') \nonumber\\
&=& \mathrm{exp}\left(\frac{1}{2}\cdot \Big(\mathbf{D}_\ast(\mathbf{x}, \mathbf{x}') - \mathbf{D}(\mathbf{x}, \mathbf{x}')\Big)\right) \label{eq:a2}
\end{eqnarray}
which immediately implies $\mathbf{D}_\ast(\mathbf{x}, \mathbf{x}') - \mathbf{D}(\mathbf{x}, \mathbf{x}') \ \leq\ 2\log(1 + \epsilon)$. Likewise, repeating the same exercise for $k(\mathbf{x},\mathbf{x}')/k_\ast(\mathbf{x},\mathbf{x}') \ \geq\ 1 - \epsilon$ implies $\mathbf{D}_\ast(\mathbf{x}, \mathbf{x}') - \mathbf{D}(\mathbf{x}, \mathbf{x}') \ \geq\ -2\log(1/(1 - \epsilon))$. Thus, combining the above, it follows that 
\begin{eqnarray}
\hspace{-18mm}\Bigg|\mathbf{D}(\mathbf{x}, \mathbf{x}') - \mathbf{D}_\ast(\mathbf{x},\mathbf{x}')\Bigg| &\leq& 2\log\left(\frac{1}{1 - \epsilon}\right) \label{eq:a3}
\end{eqnarray}
which holds because for $\epsilon \in (0, 1)$, we have $1/(1 \ -\ \epsilon) \ \geq\ 1 \ +\ \epsilon$.
\end{proof}

\section{Proof of Lemma~\ref{lem:2}}
\label{app:b}
{\bf Lemma 3}
Suppose $\displaystyle \underset{(\mathbf{x},\mathbf{x}')}{\mathrm{sup}}\ \left|\frac{k(\mathbf{x},\mathbf{x}') - k_\ast(\mathbf{x}, \mathbf{x}')}{k_\ast(\mathbf{x},\mathbf{x}')}\right| \ \leq\ \epsilon$ for $\epsilon \in (0, 1)$ then
\begin{eqnarray}
\hspace{-7mm}\forall (\mathbf{x}, \mathbf{x}'):\Big|\mathbf{D}(\mathbf{x}, \mathbf{x}') - \mathbf{D}_\ast(\mathbf{x},\mathbf{x}')\Big| \hspace{-2mm}&\leq&\hspace{-2mm} 2\log\left(\frac{1}{1 - \epsilon}\right) \label{eq:b1}
\end{eqnarray}

\begin{proof}
$\displaystyle \underset{(\mathbf{x},\mathbf{x}')}{\mathrm{sup}}\ \left|\frac{k(\mathbf{x},\mathbf{x}') - k_\ast(\mathbf{x}, \mathbf{x}')}{k_\ast(\mathbf{x},\mathbf{x}')}\right| \ \leq\ \epsilon$ implies
\begin{eqnarray}
\hspace{-8mm}\forall (\mathbf{x},\mathbf{x}'): k_\ast\left(\mathbf{x},\mathbf{x}'\right)\left(1 \ -\ \epsilon\right) \hspace{-2mm}&\leq&\hspace{-2mm}  k\left(\mathbf{x},\mathbf{x}'\right)\nonumber\\
\hspace{-2mm}&\leq&\hspace{-2mm} k_\ast\left(\mathbf{x},\mathbf{x}'\right)\left(1 \ +\  \epsilon\right)  \label{eq:b2}
\end{eqnarray}
Thus, by Lemma~\ref{lem:1}, Eq.~\eqref{eq:b2} subsequently implies:
\begin{eqnarray}
\hspace{-8mm}\forall (\mathbf{x},\mathbf{x}'): \Big|\mathbf{D}(\mathbf{x}, \mathbf{x}') - \mathbf{D}_\ast(\mathbf{x},\mathbf{x}')\Big| \hspace{-2mm}&\leq&\hspace{-2mm} 2\log\left(\frac{1}{1 - \epsilon}\right)\label{eq:b3}
\end{eqnarray}
\end{proof}

\section{Proof of Theorem~\ref{theo:1}}
\label{app:c}
To prove Theorem~\ref{theo:1}, we need to first establish the following auxiliary results:\\

\begin{lemma}
\label{lem:3}
Let $\mathbf{r} \sim \mathbb{N}(0, \mathbf{K}_\ast)$ and $\mathbf{A}$ denote a positive semi-definite and symmetric matrix. We have:
\begin{eqnarray}
\hspace{-8mm}\forall \lambda > 0:\ \mathbb{E}\left[\mathrm{exp}\left(\lambda\mathbf{r}^\top\mathbf{A}\mathbf{r}\right)\right] &=& \left|\mathbf{I} - 2\lambda\mathbf{A}\mathbf{K}_\ast\right|^{-\frac{1}{2}} \label{eq:c1}
\end{eqnarray}
where the expectation is over the distribution of $\mathbf{r}$. 
\end{lemma}

\begin{proof}
Note that for any $\mathbf{r} \sim \mathbb{N}(0, \mathbf{B})$ where $\mathbf{B}$ is a symmetric, positive semi-definite matrix,
\begin{eqnarray}
\hspace{-12mm}p(\mathbf{r}) &=& \left(2\pi\right)^{-\frac{n}{2}} \left|\mathbf{B}\right|^{-\frac{1}{2}} \mathrm{exp}\left(-\frac{1}{2}\mathbf{r}^\top\mathbf{B}^{-1}\mathbf{r}\right) \ .\label{eq:c2}
\end{eqnarray}
This also implies
\begin{eqnarray}
\hspace{-19mm}\int_{\mathbf{r}}\mathrm{exp}\left(-\frac{1}{2}\mathbf{r}^\top\mathbf{B}^{-1}\mathbf{r}\right)\mathrm{d}\mathbf{r} &=& \left(2\pi\right)^{\frac{n}{2}} \left|\mathbf{B}\right|^{\frac{1}{2}} \label{eq:c3}
\end{eqnarray}
since $p(\mathbf{r})$ must integrate to one. On the other hand, we can expand $\mathbb{E}\left[\mathrm{exp}\left(\lambda\mathbf{r}^\top\mathbf{A}\mathbf{r}\right)\right]$ 
\begin{eqnarray}
\hspace{-3.5mm}&=&\hspace{-3mm} \left(2\pi\right)^{-\frac{n}{2}}\left|\mathbf{K}_\ast\right|^{-\frac{1}{2}}\hspace{-2mm}\int_{\mathbf{r}}\mathrm{exp}\left(\lambda\mathbf{r}^\top\mathbf{A}\mathbf{r}\right)\mathrm{exp}\left(-\frac{1}{2}\mathbf{r}^\top\mathbf{K}_\ast^{-1}\mathbf{r}\right)\mathrm{d}\mathbf{r} \nonumber\\
\hspace{-3.5mm}&=&\hspace{-3mm} \left(2\pi\right)^{-\frac{n}{2}}\left|\mathbf{K}_\ast\right|^{-\frac{1}{2}}\hspace{-2mm}\int_{\mathbf{r}}\mathrm{exp}\left(-\frac{1}{2}\mathbf{r}^\top\left(\mathbf{K}_\ast^{-1} - 2\lambda\mathbf{A}\right)\mathbf{r}\right)\mathrm{d}\mathbf{r}\nonumber\\
\hspace{-3.5mm}&=&\hspace{-3mm} \left(2\pi\right)^{-\frac{n}{2}}\left|\mathbf{K}_\ast\right|^{-\frac{1}{2}}\left(2\pi\right)^{\frac{n}{2}}\left|\mathbf{B}\right|^{\frac{1}{2}} = \left|\mathbf{K}_\ast\right|^{-\frac{1}{2}}\left|\mathbf{B}\right|^{\frac{1}{2}} \label{eq:c4}
\end{eqnarray}
where $\mathbf{B} = (\mathbf{K}_\ast^{-1} - 2\lambda\mathbf{A})^{-1}$ and the second last step above follows from Eq.~\eqref{eq:c3}. Furthermore, since $\mathbf{B} = (\mathbf{K}_\ast^{-1} - 2\lambda\mathbf{A})^{-1} = ((\mathbf{I} - 2\lambda\mathbf{A}\mathbf{K}_\ast)\mathbf{K}_\ast^{-1})^{-1} = \mathbf{K}(\mathbf{I} - 2\lambda\mathbf{A}\mathbf{K}_\ast)^{-1}$, it follows that $|\mathbf{B}| = |\mathbf{K}_\ast||\mathbf{I} - 2\lambda\mathbf{A}\mathbf{K}_\ast|^{-1}$. Plugging this into Eq.~\eqref{eq:c4} yields
\begin{eqnarray}
\hspace{-1.5mm}\mathbb{E}\left[\mathrm{exp}\left(\lambda\mathbf{r}^\top\mathbf{A}\mathbf{r}\right)\right] &=& \left|\mathbf{K}_\ast\right|^{-\frac{1}{2}}\left|\mathbf{B}\right|^{\frac{1}{2}} \nonumber\\
&=&  \left|\mathbf{K}_\ast\right|^{-\frac{1}{2}}\left|\mathbf{K}_\ast\right|^{\frac{1}{2}}|\mathbf{I} - 2\lambda\mathbf{A}\mathbf{K}_\ast|^{-\frac{1}{2}} \nonumber\\
&=& |\mathbf{I} - 2\lambda\mathbf{A}\mathbf{K}_\ast|^{-\frac{1}{2}} \ .\label{eq:c5}
\end{eqnarray}
As the above is true for all $\lambda > 0$, our proof is completed.
\end{proof}

\begin{lemma}
\label{lem:4}
Let $\mathbf{r} \sim \mathbb{N}(0, \mathbf{K}_\ast)$ where $\mathbf{r} = [r(\mathbf{x}_1), r(\mathbf{x}_2), \ldots, r(\mathbf{x}_n)]$. Suppose we use $\mathbf{r}$ as observations of a surrogate feedback to fit our self-supervised GP model in Section~\ref{sec:method-ssl} and let $\widehat{\mathbf{r}}$ denote the prediction made by the fitted GP at the same set of training inputs $\{\mathbf{x}_1, \mathbf{x}_2, \ldots, \mathbf{x}_n\}$, we have
\begin{eqnarray}
\left(\mathbf{r} - \widehat{\mathbf{r}}\right)^\top\hspace{-2mm}\mathbf{A}\left(\mathbf{r} - \widehat{\mathbf{r}}\right) \hspace{-2mm}&=&\hspace{-2mm} \sigma^4\mathbf{r}^\top\hspace{-2mm}\left(\mathbf{K} + \sigma^2\mathbf{I}\right)^{-1}\hspace{-2mm}\mathbf{A}\left(\mathbf{K} + \sigma^2\mathbf{I}\right)^{-1}\mathbf{r} \nonumber\\ \label{eq:c6}
\end{eqnarray}
where $\mathbf{A}$ is a square matrix and $\sigma^2$ is the variance of the GP likelihood as defined in Eq.~\eqref{eq:2}.
\end{lemma}
\begin{proof}
Applying Eq.~\eqref{eq:2} on $\mathbf{x}_\ast = \mathbf{x}_1, \mathbf{x}_2, \ldots, \mathbf{x}_n$ and $\mathbf{y} = \mathbf{r}$ and collecting the results in a column vector $\widehat{\mathbf{r}}$ straight-forwardly yields
\begin{eqnarray}
\widehat{\mathbf{r}} &=& \mathbf{K}\left(\mathbf{K} + \sigma^2\mathbf{I}\right)^{-1}\mathbf{r} \ , \label{eq:c7}
\end{eqnarray}
which further implies
\begin{eqnarray}
\mathbf{r} - \widehat{\mathbf{r}} &=& \mathbf{r} - \mathbf{K}\left(\mathbf{K} + \sigma^2\mathbf{I}\right)^{-1}\mathbf{r} \\
&=& \mathbf{r} - \mathbf{K}\left(\mathbf{K} + \sigma^2\mathbf{I}\right)^{-1}\mathbf{r} -\sigma^2\mathbf{I}\left(\mathbf{K} + \sigma^2\mathbf{I}\right)^{-1}\mathbf{r}\nonumber\\ 
&+& \sigma^2\mathbf{I}\left(\mathbf{K} + \sigma^2\mathbf{I}\right)^{-1}\mathbf{r}\\
&=& \mathbf{r} - \left(\mathbf{K} + \sigma^2\mathbf{I}\right)\left(\mathbf{K} + \sigma^2\mathbf{I}\right)^{-1}\mathbf{r} \nonumber\\ 
&+& \sigma^2\left(\mathbf{K} + \sigma^2\mathbf{I}\right)^{-1}\mathbf{r}\\
&=& \mathbf{r} - \mathbf{r} + \sigma^2\left(\mathbf{K} + \sigma^2\mathbf{I}\right)^{-1}\mathbf{r} \nonumber\\ 
&=& \sigma^2\left(\mathbf{K} + \sigma^2\mathbf{I}\right)^{-1}\mathbf{r} \ .\label{eq:c8}
\end{eqnarray}
Finally, plugging Eq.~\eqref{eq:c8} into the expression of $\left(\mathbf{r} - \widehat{\mathbf{r}}\right)^\top\mathbf{A}\left(\mathbf{r} - \widehat{\mathbf{r}}\right)$ produces the desired result.
\end{proof}

Given the above results in Lemma~\ref{lem:3} and Lemma~\ref{lem:4}, we are now ready to (re-)state and prove the key result in Theorem~\ref{theo:1} below.

{\bf Theorem 1}
Let $\displaystyle g(\tau) \triangleq \log(\tau) + (1/\tau) - 1$ and $\displaystyle c_{\epsilon} \triangleq \epsilon\lambda_{\max}\Big(\mathbf{K}_\ast\Big)/d$, we have
\begin{eqnarray}
\hspace{-10mm}&\ &\mathrm{Pr}\Bigg(\underset{(\mathbf{x},\mathbf{x}')}{\mathrm{sup}}\ \left|\frac{k(\mathbf{x},\mathbf{x}') - k_\ast(\mathbf{x},\mathbf{x}')}{k_\ast(\mathbf{x},\mathbf{x}')}\right| \ \geq\  \epsilon\Bigg)\nonumber\\ 
\hspace{-10mm}&\leq& \mathrm{exp}\left(-\frac{1}{2}n \cdot g\left(\frac{\sigma^4}{\alpha}\Big(1 - c_{\epsilon}\Big)\lambda_{\max}\Big(\mathbf{K}_\ast\Big)\right)\right) \label{eq:c9}
\end{eqnarray}
where $d$ and $\alpha$ are defined in {\bf A1} and {\bf A2} above. 
\begin{proof}
To prove the above result, note that
\begin{eqnarray}
\hspace{-10mm}&\ &\hspace{-2mm}\mathrm{Pr}\left(\underset{(\mathbf{x},\mathbf{x}')}{\mathrm{sup}}\ \left|\frac{k(\mathbf{x},\mathbf{x}') - k_\ast(\mathbf{x},\mathbf{x}')}{k_\ast(\mathbf{x},\mathbf{x}')}\right| \ \geq \ \epsilon\right) \nonumber\\
\hspace{-10mm}&=&\hspace{-2mm} \mathrm{Pr}\Bigg(\Big|k(\mathbf{x}_a,\mathbf{x}_b) - k_\ast(\mathbf{x}_a,\mathbf{x}_b)\Big| \ \geq \ \epsilon \cdot k_\ast(\mathbf{x}_a,\mathbf{x}_b)\Bigg) \label{eq:c10}
\end{eqnarray}
where $(\mathbf{x}_a,\mathbf{x}_b)$ is any pair of inputs for which the corresponding multiplicative approximation error meets the defined supremum value. Thus, let $\mathbb{E}$ be the event that $\Big|k(\mathbf{x}_a,\mathbf{x}_b) - k_\ast(\mathbf{x}_a,\mathbf{x}_b)\Big| \ \geq \ \epsilon \cdot k_\ast(\mathbf{x}_a,\mathbf{x}_b)$. By assumptions {\bf A1} and {\bf A2}, we have $\mathrm{Pr}(\mathbb{E})$
\begin{eqnarray}
\hspace{-6mm}&\leq&\hspace{-2mm} \mathrm{Pr}\Bigg(\underset{(\mathbf{x},\mathbf{x}')}{\mathrm{sup}}\ \Big|k(\mathbf{x},\mathbf{x}') - k_\ast(\mathbf{x},\mathbf{x}')\Big| \ \geq \ \epsilon \cdot k_\ast(\mathbf{x}_a,\mathbf{x}_b)\Bigg)\nonumber\\
\hspace{-6mm}&\leq&\hspace{-2mm} \mathrm{Pr}\Bigg(\underset{(\mathbf{x},\mathbf{x}')}{\mathrm{sup}}\ \Big|k(\mathbf{x},\mathbf{x}') - k_\ast(\mathbf{x},\mathbf{x}')\Big| \ \geq \ \epsilon \cdot \frac{\lambda_{\max}\left(\mathbf{K}_\ast\right)}{d}\Bigg)\nonumber\\
\hspace{-6mm}&\leq&\hspace{-2mm} \mathrm{Pr}\Bigg(\frac{(\mathbf{r} - \widehat{\mathbf{r}})^\top\mathbf{A}(\mathbf{r} - \widehat{\mathbf{r}})-\alpha}{(\mathbf{r} - \widehat{\mathbf{r}})^\top\mathbf{A}(\mathbf{r} - \widehat{\mathbf{r}})} \ \geq \ c_{\epsilon}\Bigg) \nonumber\\
\hspace{-6mm}&=&\hspace{-2mm} \mathrm{Pr}\Bigg((1 - c_{\epsilon})(\mathbf{r} - \widehat{\mathbf{r}})^\top\mathbf{A}(\mathbf{r} - \widehat{\mathbf{r}}) \ \geq \ \alpha\Bigg) \label{eq:c11}
\end{eqnarray}
with $c_{\epsilon} = \epsilon \cdot \lambda_{\max}(\mathbf{K}_\ast) / d$ and $\mathbf{A} = (1/n)(\mathbf{K} + \sigma^2\mathbf{I})^2$ as defined in {\bf A2}. Here, the second and third inequalities in the above follow from {\bf A1} and {\bf A2}, respectively. Next, applying Lemma~\ref{lem:4} to the RHS of Eq.~\eqref{eq:c11} with $\mathbf{A} = (1/n)(\mathbf{K} + \sigma^2\mathbf{I})^2$, it follows that $\forall \lambda > 0$:
\begin{eqnarray}
\mathrm{Pr}\Big(\mathbb{E}\Big) \hspace{-2.5mm}&\leq&\hspace{-2.5mm} \mathrm{Pr}\Bigg((1 - c_{\epsilon})(\mathbf{r} - \widehat{\mathbf{r}})^\top\mathbf{A}(\mathbf{r} - \widehat{\mathbf{r}}) \geq \alpha\Bigg) \nonumber\\
\hspace{-2.5mm}&=&\hspace{-2.5mm} \mathrm{Pr}\Bigg(\frac{\sigma^4}{n}(1 - c_{\epsilon})\mathbf{r}^\top\mathbf{r} \geq \alpha\Bigg)\nonumber\\
\hspace{-2.5mm}&=&\hspace{-2.5mm} \mathrm{Pr}\Bigg(\mathrm{exp}\left(\lambda \cdot \frac{\sigma^4}{n}(1 - c_{\epsilon})\mathbf{r}^\top\mathbf{r}\right) \geq \mathrm{exp}\left(\lambda \cdot \alpha\right) \Bigg)\nonumber\\
\hspace{-2.5mm}&\leq&\hspace{-2.5mm} \mathbb{E}\left[\mathrm{exp}\left(\lambda \cdot \frac{\sigma^4}{n}(1 - c_{\epsilon})\mathbf{r}^\top\mathbf{r}\right)\right]\mathrm{exp}\left(-\lambda \cdot \alpha\right)\nonumber\\
\hspace{-2.5mm}&=&\hspace{-2.5mm} \left|\mathbf{I} - 2\lambda \cdot \frac{\sigma^4}{n}(1 - c_{\epsilon})\mathbf{K}_\ast\right|^{-\frac{1}{2}}\mathrm{exp}\left(-\lambda \cdot \alpha\right)\nonumber\\
\hspace{-2.5mm}&=&\hspace{-2.5mm} \mathrm{exp}\left(\hspace{-1mm}-\frac{1}{2}\log\left|\mathbf{I} - 2\lambda \cdot \frac{\sigma^4}{n}(1 - c_{\epsilon})\mathbf{K}_\ast\right| \hspace{-1mm}-\hspace{-1mm} \lambda \cdot \alpha\right) \nonumber\nonumber\\
\hspace{-2.5mm}&=&\hspace{-2.5mm}\mathrm{exp}\left(\mathrm{F}(\lambda)\right) \label{eq:c12}
\end{eqnarray}
where $\mathrm{F}(\lambda) = -\frac{1}{2}\log\left|\mathbf{I} - 2\lambda \cdot \frac{\sigma^4}{n}(1 - c_{\epsilon})\mathbf{K}_\ast\right| - \lambda \cdot \alpha$. Here, the second step above follows from Lemma~\ref{lem:4}. The fourth step follows from the Markov inequality and the fifth step results from applying Lemma~\ref{lem:3} when we replace $\lambda$ (in Lemma~\ref{lem:3}) by $\lambda \cdot \frac{\sigma^4}{n}(1 - c_{\epsilon})$ and $\mathbf{A} = \mathbf{I}$. 

Furthermore, we also note that Eq.~\eqref{eq:c12} holds for all $\lambda > 0$, we can tighten the bound on the RHS by solving for $\lambda$ that minimizes $\mathrm{F}(\lambda)$. To do this, we set $\mathrm{d}\mathrm{F}(\lambda)/\mathrm{d}\lambda = 0$ and solve for $\lambda$ which results in
\begin{eqnarray}
\lambda &=& \frac{n}{2}\cdot\frac{1}{\sigma^4}\cdot\frac{1}{1 - c_{\epsilon}}\cdot\frac{1}{\lambda_{\max}(\mathbf{K}_\ast)}\nonumber\\
&\times&\left(1 - \frac{\sigma^4}{\alpha} \cdot (1 - c_{\epsilon})\cdot \lambda_{\max}(\mathbf{K}_\ast)\right) \label{eq:c13}
\end{eqnarray}
Thus, plugging this optimal value of $\lambda$ into the RHS of Eq.~\eqref{eq:c12} yields 
\begin{eqnarray}
\hspace{-7mm}\mathrm{Pr}\Big(\mathbb{E}\Big) \hspace{-2mm}&\leq&\hspace{-2mm} \mathrm{exp}\left(-\frac{n}{2} \cdot g\left(\frac{\sigma^4}{\alpha}(1 - c_{\epsilon})\lambda_{\max}\Big(\mathbf{K}_\ast\Big)\right)\right) \label{eq:c14}
\end{eqnarray}
where $g(\tau) = \log(\tau) + (1/\tau) - 1$. From Eq.~\eqref{eq:c10} and the definition of $\mathbb{E}$, we have
\begin{eqnarray}
\hspace{-12mm}\mathrm{Pr}\left(\mathrm{sup}\ \left|\frac{k(\mathbf{x},\mathbf{x}') - k_\ast(\mathbf{x},\mathbf{x}')}{k_\ast(\mathbf{x},\mathbf{x}')}\right| \geq \epsilon\right) \hspace{-2mm}&=&\hspace{-2mm} \mathrm{Pr}\Big(\mathbb{E}\Big) \label{eq:c15}
\end{eqnarray}
Then, combining this with Eq.~\eqref{eq:c14} above yields the desired result.
\end{proof}

\section{Proof of Theorem~\ref{theo:2}}
\label{app:d}
{\bf Theorem 2}
Let $\displaystyle g(\tau) = \log(\tau) + (1/\tau) - 1$ and $\displaystyle g_{\epsilon} = g\left(\frac{\sigma^4}{\alpha}\left(1 - \lambda_{\max}\Big(\mathbf{K}_\ast\Big)\frac{\epsilon}{d}\right)\lambda_{\max}\Big(\mathbf{K}_\ast\Big)\right)$. Then,
\begin{eqnarray}
\mathrm{Pr}\Bigg(\underset{(\mathbf{x},\mathbf{x}')}{\mathrm{sup}}\ \Big|\mathbf{D}(\mathbf{x},\mathbf{x}') - \mathbf{D}_\ast(\mathbf{x},\mathbf{x}')\Big| &\leq& 2\log\left(\frac{1}{1 - \epsilon}\right)\Bigg) \nonumber\\
&\geq& 1 - \delta \label{eq:d1}
\end{eqnarray}
when $\displaystyle n \geq \frac{2}{g_{\epsilon}}\log\frac{1}{\delta}$ and $\delta \in (0, 1)$ is an arbitrarily small confidence parameter. 

\begin{proof}
Setting $\mathrm{exp}\left(-\frac{n}{2} \cdot g\left(\frac{\sigma^4}{\alpha}(1 - c_{\epsilon})\lambda_{\max}\Big(\mathbf{K}_\ast\Big)\right)\right) \ \leq\ \delta$ implies
\begin{eqnarray}
\hspace{-12mm}\mathrm{Pr}\left(\underset{(\mathbf{x},\mathbf{x}')}{\mathrm{sup}}\ \left|\frac{k(\mathbf{x},\mathbf{x}') - k_\ast(\mathbf{x},\mathbf{x}')}{k_\ast(\mathbf{x},\mathbf{x}')}\right| \ \geq \ \epsilon\right) &\leq& \delta \label{eq:d2}
\end{eqnarray}
via Theorem~\ref{theo:1} or equivalently,
\begin{eqnarray}
\hspace{-9.5mm}\mathrm{Pr}\left(\underset{(\mathbf{x},\mathbf{x}')}{\mathrm{sup}}\ \left|\frac{k(\mathbf{x},\mathbf{x}') - k_\ast(\mathbf{x},\mathbf{x}')}{k_\ast(\mathbf{x},\mathbf{x}')}\right| \ \leq \ \epsilon\right) \hspace{-2mm}&\geq&\hspace{-2mm} 1 - \delta \label{eq:d3}
\end{eqnarray}
That is, with probability at least $1 - \delta$, 
\begin{eqnarray}
\underset{(\mathbf{x},\mathbf{x}')}{\mathrm{sup}}\ \left|\frac{k(\mathbf{x},\mathbf{x}') - k_\ast(\mathbf{x},\mathbf{x}')}{k_\ast(\mathbf{x},\mathbf{x}')}\right| &\leq& \epsilon \label{eq:d4}
\end{eqnarray}
which in turn implies
\begin{eqnarray}
\hspace{-11mm}\underset{(\mathbf{x},\mathbf{x}')}{\mathrm{sup}}\ \Big|\mathbf{D}(\mathbf{x},\mathbf{x}') - \mathbf{D}_\ast(\mathbf{x}, \mathbf{x}')\Big| &\leq& 2\log\left(\frac{1}{1 - \epsilon}\right) \label{eq:d5}
\end{eqnarray}
via Lemma~\ref{lem:2}. This also means
\begin{eqnarray}
\mathrm{Pr}\Bigg(\underset{(\mathbf{x},\mathbf{x}')}{\mathrm{sup}}\ \Big|\mathbf{D}(\mathbf{x},\mathbf{x}') - \mathbf{D}_\ast(\mathbf{x},\mathbf{x}')\Big| &\leq& 2\log\left(\frac{1}{1 - \epsilon}\right)\Bigg) \nonumber\\
&\geq& 1 - \delta \ .\label{eq:d6}
\end{eqnarray}
As this happens when 
\begin{eqnarray}
\hspace{-10.2mm}\mathrm{exp}\left(-\frac{n}{2} \cdot g\left(\frac{\sigma^4}{\alpha}(1 - c_{\epsilon})\lambda_{\max}\Big(\mathbf{K}_\ast\Big)\right)\right) &\leq& \delta \ ,
\end{eqnarray}
we have
\begin{eqnarray}
\hspace{-9mm}\log\frac{1}{\delta} \hspace{-2mm}&\leq&\hspace{-2mm} \frac{n}{2} \cdot g\left(\frac{\sigma^4}{\alpha}(1 - c_{\epsilon})\lambda_{\max}\Big(\mathbf{K}_\ast\Big)\right) = \frac{n}{2} \cdot g_{\epsilon} \label{eq:d7}
\end{eqnarray}
which implies $n$ must be at least $\displaystyle\frac{2}{g_{\epsilon}}\log\frac{1}{\delta}$ as desired. 
\end{proof}

\section{Sparse Gaussian Processes}
\label{app:e}
To improve the scalability of GP model, numerous sparse GP methods \citep{NghiaAAAI17,Hoang2020,NghiaAAAI19,NghiaAAAI18,Miguel10,Candela07,Snelson07,Titsias09,NghiaIJCNN19} have been proposed to reduce its cubic processing cost to linear in the size of data. One common recipe that is broadly adopted by these works is the exploitation of a low-rank approximate representation of the covariance (or Gram) matrix in characterizing the GP prior. 

The work of \citet{Candela05} has in fact presented a unifying view of such methods, which share a similar structural assumption of conditional independence (albeit of varying degrees) based on a notion of inducing variables \citep{Tresp03,Seeger03,Smola01,Snelson06,Snelson07a}. In particular, under this assumption, there exists a subset of supporting inputs $\mathbb{S} = \{\mathbf{x}_1, \mathbf{x}_2, \ldots, \mathbf{x}_m\}$ such that conditioning on their (latent) outputs $\mathbf{g}_+ = [g(\mathbf{x}_1)) \ldots g(\mathbf{x}_m)]^\top$, the input space can be partitioned into regions such that if $\mathbf{x}$ and $\mathbf{x}'$ belong to two different regions, then $g(\mathbf{x})$ and $g(\mathbf{x}')$ are statistically independent.

Depending on the specific form the assumed conditional independence, which entails how the input space is partitioned and is different across different methods, the exact covariance matrix $\mathbf{K}$ can be shown to be equivalent to either $\mathbf{Q}$ \citep{Seeger03,Smola01}, $\mathbf{Q} - \mathrm{diag}[\mathbf{Q} - \mathbf{K}]$ \cite{Snelson06} or $\mathbf{Q} - \mathrm{blkdiag}[\mathbf{Q} - \mathbf{K}]$ \citep{Tresp03,Snelson07} (among others). Here, common to all these approximations is the low-rank matrix $\mathbf{Q}$ whose entries $\mathbf{Q}_{ab} \triangleq \mathbf{k}_a^\top \mathbf{K}_{++}^{-1} \mathbf{k}_b$ where $\mathbf{k}_{a} \triangleq [k(\mathbf{x}_a, \mathbf{x}_1) \ldots k(\mathbf{x}_a, \mathbf{x}_m)]^\top$, $\mathbf{k}_{b} \triangleq [k(\mathbf{x}_b, \mathbf{x}_1) \ldots k(\mathbf{x}_b, \mathbf{x}_m)]^\top$ and $\mathbf{K}_{++}$ is the Gram matrix induced by $k(\mathbf{x}, \mathbf{x}')$ on $\mathbb{S} = \{\mathbf{x}_1, \mathbf{x}_2, \ldots, \mathbf{x}_m\}$. 

By its construction, it is straight-forward to see that $\mathrm{rank}(\mathbf{Q}) \leq m$ and as such, plugging $\mathbf{Q}$ into one of the above approximations of $\mathbf{K}$ and replacing $\mathbf{K}$ in Eq.~\eqref{eq:3} with the approximation would result in an expression that is computable in $\mathbf{O}(nm^2)$ -- via the Woodburry matrix inversion identity -- which is linear in $n$. For interested readers, the exact derivation of these approximations can be found in \citep{Candela05}. In our self-supervised metric learning experiment (Section~\ref{sec:method-ssl}), we adopt the simplest approximation with $\mathbf{Q}$ which is sufficient to scale with datasets spanning tens of thousands of items.

{\bf Remark.} The approximation with $\mathbf{Q}$ \citep{Smola01} is later rediscovered as the result performing variational inference to approximate the tractable (but otherwise costly) GP prediction \citep{Titsias09}. This results in an alternative variational perspective that unifies a broader spectrum (including the above) of sparse GPs. Under this new view, the above approximations can be reproduced as the exact result of performing variational inference on a GP with modified prior \citep{NghiaICML15} or noise covariance \citep{NghiaICML16}.

\end{document}